\documentclass[11pt]{article}
\usepackage{geometry}                
\usepackage[parfill]{parskip}    

\usepackage[english]{babel}

\usepackage[pdftex]{graphicx,color}
\usepackage{amssymb,amsmath,amsthm}
\usepackage{epstopdf}
\usepackage{mathtools}
\usepackage{empheq}
\usepackage{marvosym}
\usepackage{multirow}

\usepackage[pdftex,bookmarks,colorlinks,breaklinks]{hyperref}
\definecolor{dullmagenta}{rgb}{0.4,0,0.4}   
\definecolor{darkblue}{rgb}{0,0,0.4}
\hypersetup{linkcolor=red,citecolor=blue,filecolor=dullmagenta,urlcolor=darkblue}

\newcommand{\nn}{\nonumber}

\newcommand{\ee}{{\rm e}}
\newcommand{\dd}{{\rm d}}

\newcommand{\NN}{{\mathbb N}}

\newcommand{\CC}{{\mathbb C}}
\newcommand{\RR}{{\mathbb R}}
\newcommand{\ZZ}{{\mathbb Z}}

\newcommand{\DD}{{\mathbb D}}

\newcommand{\cc}[1]{\overline{#1}}

\def\clos{\mathop{\rm clos}}

\def\erfc{\mathop{\rm erfc}}

\def\logcap{\mathop{\rm cap}}

\newtheorem{thm}{Theorem}[section] 
\newtheorem{prop}[thm]{Proposition}
\newtheorem{lemma}[thm]{Lemma}
\newtheorem{cor}[thm]{Corollary}
\newtheorem*{thm-others}{Theorem}

\theoremstyle{remark}

\newtheorem*{remark}{Remark}

\title{Fine asymptotic behavior in eigenvalues 
\\of random normal matrices: Ellipse Case}
\author{Seung-Yeop Lee\footnotemark[1] ~and Roman Riser\footnotemark[2]}

\date{}

\begin{document}

\maketitle
\renewcommand{\thefootnote}{\fnsymbol{footnote}}
\footnotetext[1]{Department of Mathematics, University of South Florida, 4202 East Fowler Avenue, Tampa, FL 33620, USA. $\qquad $ \mbox{E-mail: lees3\symbol{'100}usf.edu}}
\footnotetext[2]{Department of Mathematics, KU Leuven, Celestijnenlaan 200B, 3001 Leuven, Belgium. $\qquad\qquad\qquad\qquad\qquad\qquad $ \mbox{E-mail: roman.riser\symbol{'100}wis.kuleuven.be}}

\begin{abstract} 
We consider the random normal matrices with quadratic external potentials where the associated orthogonal polynomials are Hermite polynomials and the limiting support (called droplet) of the eigenvalues is an ellipse.  
We calculate the density of the eigenvalues near the boundary of the droplet up to the second subleading corrections and express the subleading corrections in terms of the curvature of the droplet boundary.
From this result we additionally get the expected number of eigenvalues outside the droplet. We also obtain the asymptotics of the kernel and found that, in the bulk, the correction term is exponentially small. This leads to the vanishing of certain Cauchy transform of the orthogonal polynomial in the bulk of the droplet up to an exponentially small error. 
\end{abstract}

\tableofcontents

\section{Introduction and results}

Consider the set of $n$ point particles, $\{z_1,\cdots,z_n\}$, in the complex plane $\CC$ that interact via two-dimensional (i.e. logarithmic) Coulomb repulsion and are subject to the quadratic confining potential 
\begin{equation}\label{eq:V}
V(z) = |z|^2 - t\, {\rm Re}(z^2), \quad 0 \leq t<1.
\end{equation}
The total (electrostatic) energy of the system is given by
$$E_n\big(\{z_j\}_{j=1}^n\big)=\sum_{k=1}^n V(z_k) + \frac{2}{N}\sum_{j< k}\log\frac{1}{|z_j-z_k|}.$$
We consider the canonical ensemble of the particle system, i.e., we assign the Gibbs measure, 
\begin{equation}\label{Eq:gibbs}
\begin{split}
{\rm Prob}_n\big(\{z_j\}_{j=1}^n\big)\prod_{j=1}^n \dd A(z_j)  &= \frac{1}{{\cal Z}_n} \exp\big(-N E_n\big(\{z_j\}_{j=1}^n\big)\big)\prod_{j=1}^n \dd A(z_j)
\\ &= \frac{1}{{\cal Z}_n} \bigg(\prod_{l<k}|z_l-z_k|^2\bigg)\ee^{-N \sum_{j=1}^nV(z_j)}\prod_{j=1}^n \dd A(z_j),
\end{split}
\end{equation}
on the configuration space of the $n$ particles. 
Above, $N$ is a positive constant, $\dd A(z)$ is the two-dimensional area measure, and ${\cal Z}_n$ is the normalization constant.  

The corresponding random process is a special case of {\em Coulomb gas} or {\em $\beta$-ensemble} (see, for instance, \cite{forrester2010log}).
It also describes the eigenvalues of a random normal matrix ensemble \cite{felder-elbau,Hedenmalm01042013,Teodorescu2005407}.
The case of $t=0$ is called {\em Ginibre ensemble} and has been studied in \cite{ginibre,forrester_honner,rider_virag}. The case of non-zero $t$ has been studied, for instance in \cite{bender}. 

The {\em (averaged) density}, $\rho_n$, of the particles, is given by
\begin{equation}\label{eq:first-density}
\rho_n(z):= {\mathbb E}\bigg(\frac{1}{N}\sum_{j=1}^n\delta(z-z_j)\bigg),\end{equation}  
where the delta function is with respect to the area measure, and the average is taken with respect to the Gibbs measure \eqref{Eq:gibbs}. Note that we normalize with $1/N$ such that the total mass of $\rho_n$ is $n/N$ instead of one. In fact, throughout this paper, we consider the scaling limit where the total mass of $\rho_n$ is fixed, i.e., for a fixed $T>0$, $N$ grows with $n$ such that
\begin{equation*}
  T=\frac{n}{N}.
\end{equation*}
\begin{figure}[t!]
\begin{center}
\includegraphics[width=0.45\textwidth]{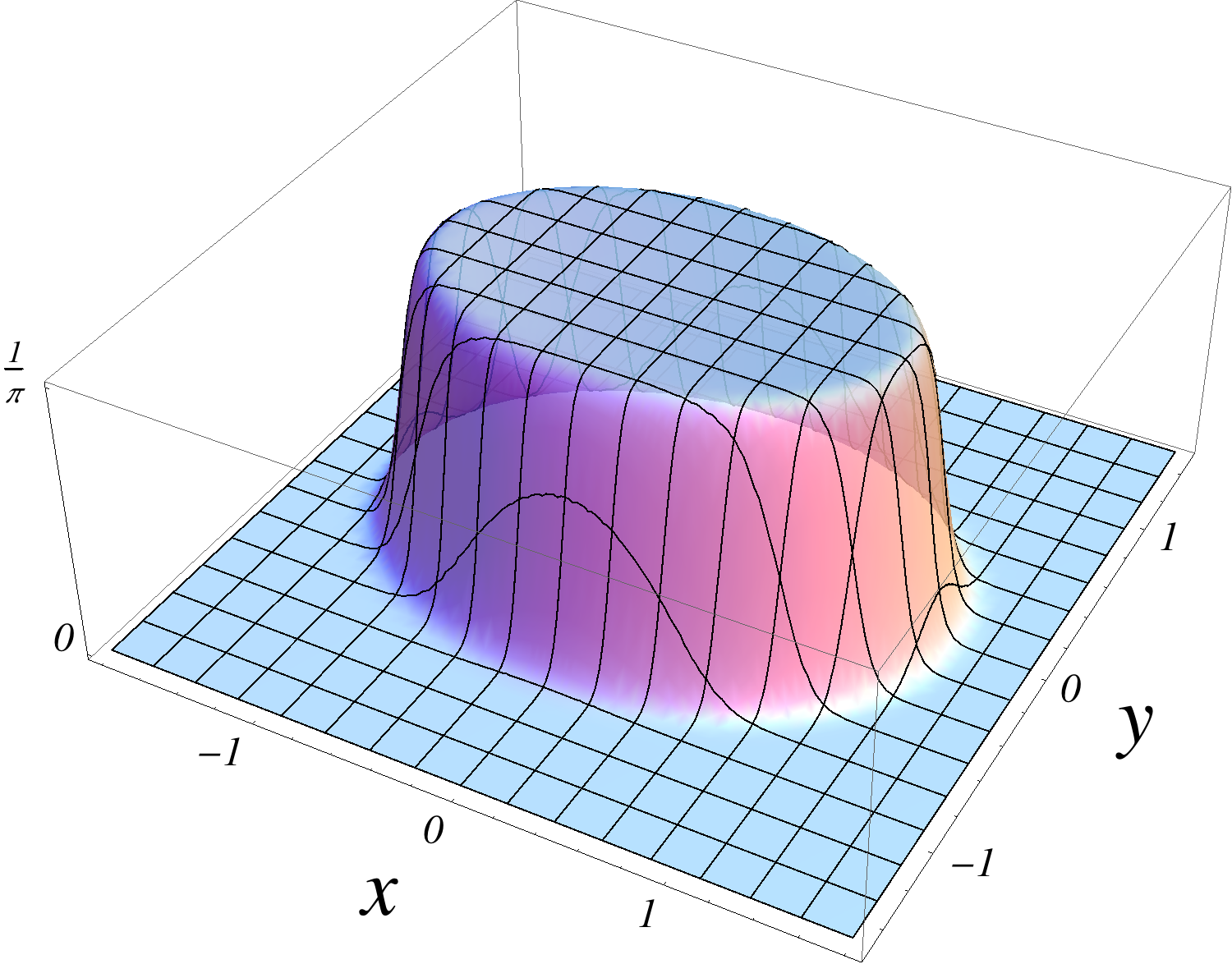}
\caption{$\rho_n$ for $t=0.25$ and $n=N=50$}\label{fig-rho}
\end{center}
\end{figure}
In Figure \ref{fig-rho}, one observes that the density is supported approximately on an elliptical shape.  In fact, as $n$ grows to infinity, the density $\rho_n$ converges  to ($1/\pi$ times) the characteristic function over 
\begin{equation}\label{eq:K} K= \bigg\{z~\bigg|~ \frac{1-t}{1+t}({\rm Re}\, z)^2 + \frac{1+t}{1-t}({\rm Im}\, z)^2\leq T\bigg\}\subset\CC. \end{equation}

\vspace{0.2cm}
\begin{remark}
The weak convergence of the density is known \cite{ameur2011,felder-elbau} for more general class of potentials where the support, $K$, of the limiting density can be an arbitrary compact set in $\CC$ depending on the potential.   In \cite{ameur2011}, $K$ is called a {\em droplet}; we will use this terminology throughout the paper.   
\end{remark}
Figure \ref{fig-rho} shows that the density changes drastically near the boundary of $K$.  Zooming at the boundary with a scaled real coordinate $\xi$ normal to the boundary, certain universal shape arises in the density profile, see Figure \ref{fig-rho2}.
The dashed line in the figure is given by \cite{forrester_honner}  
$$\lim_{n\to\infty} \rho_n\Big(\sqrt{T\frac{1+t}{1-t}}+\frac{\xi}{N^{1/2}}\Big)=\frac{1}{2\pi}{\rm erfc}(\sqrt2\, \xi),$$ 
where $$\erfc(z)=\frac{2}{\sqrt{\pi}}\int_z^\infty \ee^{-w^2} \dd w. $$
In \cite{AmeurKangMakarov} the above limit is proven for other classes of potentials.  In the next theorem, we show the more detailed behavior of the density for the potential given by \eqref{eq:V}, including the subleading corrections in $1/N$-expansions.

\begin{figure}[t]\begin{center}
\includegraphics[width=0.4\textwidth]{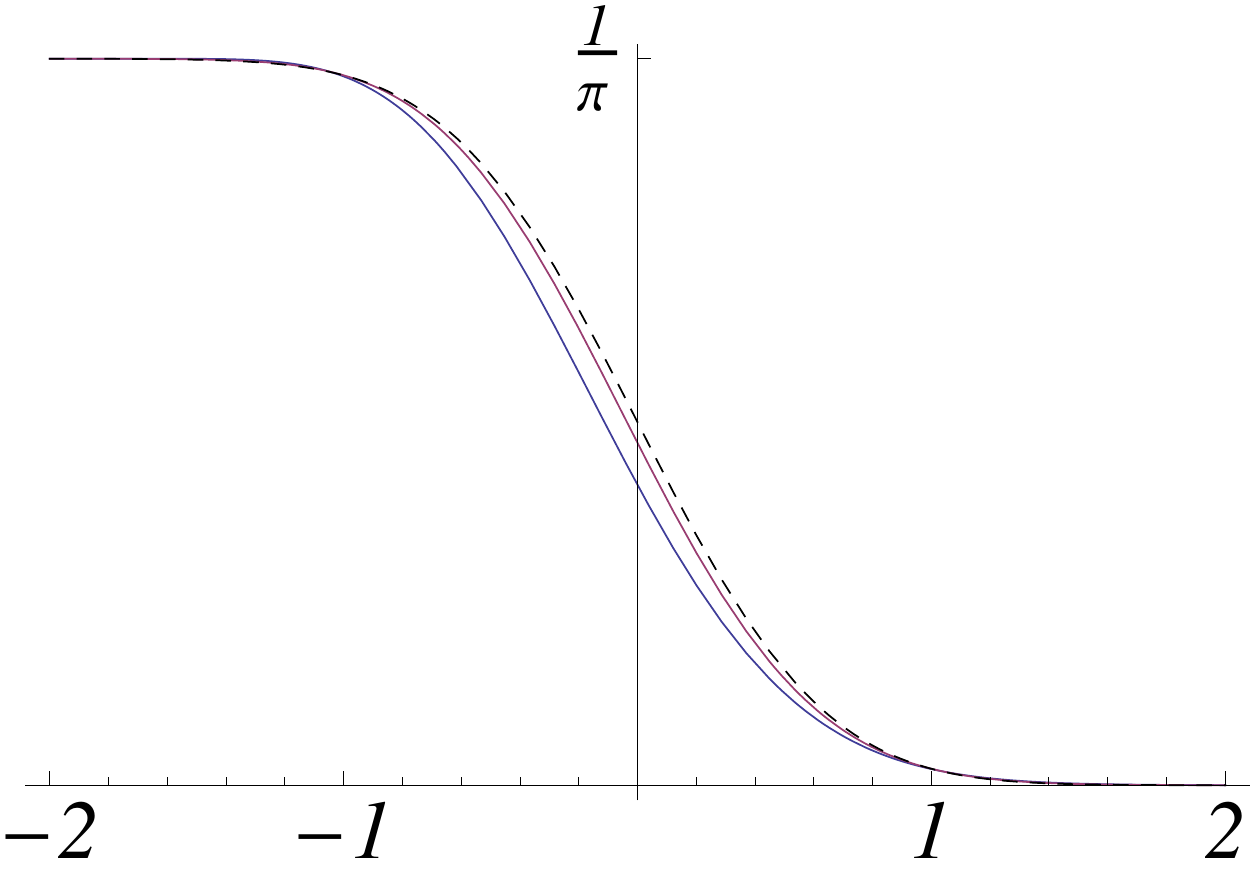}
\caption{ $\rho_n\Big(\sqrt{\frac{1+t}{1-t}}+\frac{\xi}{N^{1/2}}\Big)$ v.s. $\xi$ for  $t=0.25$, and $n=N=10,~100$ and $\infty$ (dashed line)}\label{fig-rho2} 
\end{center}
\end{figure}

\vspace{0.2cm}

\begin{thm}\label{thm1}
Given $T>0$, let the compact set $K\subset\CC$ be given by \eqref{eq:K}.
Let $z_0$ be a point on the boundary of $K$.  Let $\kappa=\kappa(z_0)$ be the (positive) curvature of $\partial K$ at the point $z_0$, and $\partial_s\kappa$ (and $\partial^2_s\kappa$) be the derivative (and second derivative respectively) of the curvature with respect to the counter-clockwise arclength parameter $s$ of the curve $\partial K$.   Let ${\bf n}={\bf n}(z_0)$ be the outward unit normal vector from $K$ at $z_0$ and $\xi$ a real coordinate.  

Let $0<\nu <1/6$. As the positive integer $n$ and the real number $N$ go to infinity while keeping $n/N=T$ fixed, we have the asymptotic expansion
\begin{equation}\label{eq:thm1}\begin{split}
\rho_{n}\bigg(z_0+\frac{\xi}{N^{1/2}}{\bf n}\bigg) = &\frac{1}{2\pi}\erfc(\sqrt{2} \xi) + \frac{\kappa}{\sqrt N}\frac{1  }{  3\sqrt{2\pi^3} } (\xi^2-1)\,\ee^{-2 \xi^2}   \\
& + \frac{1}{N}\frac{\ee^{-2 \xi^2}}{\sqrt{2\pi^3}}\left(\kappa^2 \frac{ 2\xi^5 - 8 \xi^3 + 3\xi}{18}+\left(\frac{ (\partial_s \kappa)^2 }{9 \kappa^2} - \frac{ \partial_s^2 \kappa }{12 \kappa} \right) \xi\right)   + \mathcal{O}\left(N^{-3/2+9\nu}\right),
\end{split}
\end{equation}
where the error bound is uniform over $|\xi|<N^{\nu}$ and over $z_0\in\partial K$.   
\end{thm}

In Appendix \ref{app:plots}, we show a few numerical plots regarding the subleading corrections of $\rho_n$.
For explicit expressions of $\kappa$, $\partial_s\kappa$, and $\partial^2_s\kappa$, see \eqref{eq:kappa0}, \eqref{eq:kappa1}, and \eqref{eq:kappa2}.  

\vspace{0.2cm}
\begin{remark}
The reader may wonder why we choose to write the expansion \eqref{eq:thm1} in terms of the curvature.   There are two reasons: ``simplicity'' and ``possible universality''. 
In fact, we conjecture that the term of the order $1/\sqrt N$ in \eqref{eq:thm1} is universally true for a general droplet with a real analytic boundary.  The term of order $1/N$ does not seem to hold in general and, therefore, is chosen purely for the simplicity of the expression.
\end{remark}
\vspace{0.3cm}

\begin{thm}\label{thm:1b}
Let $0< \tau\le 1/2$. Given $T>0$ there exists constants $C_0>0$ and $N_0>0$ such that for all $N>N_0$ and  $n=NT$, we have 
\begin{align*}
 \left|\rho_n(z)-\frac{1}{\pi}\right|< C_0 N^{5/6} \ee^{-N^{2\tau}} & \qquad \text{when } z\in {\rm Int}(K) \setminus U_{\tau,N},
 \\  \left|\rho_n(z)\right|< C_0 N^{5/6} \ee^{-N^{2\tau}} & \qquad \text{when } z\in {\rm Ext}(K) \setminus U_{\tau,N}.
\end{align*}
We denote the neighborhood of $\partial K$ of width $2N^{-1/2+\tau}$ by
\begin{align*}
U_{\tau,N}&=\{z_0+ X {\bf n}(z_0)\ |\ z_0\in\partial K, |X| <N^{-1/2+\tau} \} = \{z\,|\, {\rm dist}(z,\partial K) < N^{-1/2+\tau}\}. 
\end{align*}
\end{thm}

Choosing $0<\tau <\nu<1/6$, Theorem \ref{thm1} and Theorem \ref{thm:1b} together give a uniform estimate of $\rho_n$ over the whole complex plane.

It is also interesting to look at the kernel function defined by (see, for instance, \cite{forrester2010log})
\begin{equation}\label{eq:Kernel}
K_n(z,w)= \sum_{j=0}^{n-1}\, p_j(z)\,\overline{p_j(w)}\,\ee^{-NV(z)/2}\ee^{-NV(w)/2},\end{equation}
where the polynomial $p_j$ of degree $j$ is the normalized orthogonal polynomial satisfying
\begin{equation}\label{eq:OP}
	\int_\CC p_j(z)\,\overline{p_k(z)}\, \ee^{-NV(z)}\dd A(z)=\delta_{jk}.
\end{equation}
The following theorem describes the asymptotic behavior of $K_n$.
\vspace{0.2cm}

\begin{thm}\label{thm2}
Given $T>0$, let the compact set $K\subset\CC$ be given by \eqref{eq:K}.
Let $V_K$ be a compact subset of ${\rm Int}\, K$ and $U_K$ a neighborhood of $K$. Then there exist $\epsilon>0$ and $\delta>0$ such that, as $n=NT$ grows to infinity, the following asymptotic expression holds uniformly over all $w$ and $z$ as specified below.
\begin{align*}
\frac{1}{N} K_{n}(z,w) \! = \! \left\{ \begin{array}{ll} \displaystyle \!\! \frac{1}{\pi} \exp\left(-\tfrac{N}{2}|z-w|^2\right)\exp\left(iN\, {\rm Im}\Big(z\cc{w}\! - \! \tfrac{t}{2}z^2 \! - \! \tfrac{t}{2}\overline w^2 \Big) \right)
\! \big(1+\mathcal{O}(\ee^{-N\epsilon})\big), & w\in V_K, z\in B_\delta(w),\vspace{0.2cm}\\
\! \mathcal{O}(\ee^{-N\epsilon}),& w\in V_K, z\notin B_\delta(w),\vspace{0.2cm}\\
\! \mathcal{O}(\ee^{-N\epsilon}),&  w\notin U_K, z\in\CC, 
\end{array}\right.
\end{align*}
The notation $B_\delta(w)$ stands for the disk of radius $\delta$ centered at $w$.  We note that the above holds when $z$ and $w$ are exchanged.
\end{thm}
\vspace{0.3cm}

Let us consider more general potential given by
\begin{equation}\label{eq:generalV}
V(z) = |z|^2- Q(z)-\overline{Q(z)}.
\end{equation}
where $Q:\CC\to\CC$ is an analytic function (in a neighborhood of $K$) such that $\lim_{|z|\to\infty }V(z)/\log|z|=+\infty$.
Then the first asymptotic formula in the above theorem can be written more generally as \eqref{eq:generalKn} below.
For a general potential the leading behavior of the kernel is known in \cite{ameur2011,Ber}.  Our new observation is that, for $Q(z)=t z^2/2$, the error bound in \eqref{eq:generalKn} is {\em exponentially small} in $N$.  We conjecture that this is true for more general $Q$.  In the next theorem, we connect the behavior of the kernel with the behavior of certain Cauchy transform.
\vspace{0.2cm}

\begin{thm}\label{thm:Cauchy}
Let $V$ be given by \eqref{eq:generalV} and let $Q$ be analytic in a neighborhood of $z\in\CC$.  Let $p_n$ and $K_n$ be the corresponding orthonormal polynomial and the kernel as defined in \eqref{eq:OP} and \eqref{eq:Kernel} respectively. 
In the scaling limit where $n=NT$ goes to $\infty$ for a fixed $T>0$, let $K_n$ satisfy the asymptotic behavior:  
\begin{equation}\label{eq:generalKn}
 \frac{1}{N} K_{n}(z,w)= \displaystyle\frac{1}{\pi} 
 e^{-\frac{N}{2}|z-w|^2}e^{iN\, {\rm Im}\big(z\cc{w}-Q(z)- \overline{Q(w)} \big) }
\big(1+\mathcal{O}(\ee^{-N\epsilon})\big),
\end{equation}
uniformly for $w\in B_\delta(z)$ for some $\epsilon>0$ and $\delta>0$.  Then there exists $\epsilon'>0$ such that 
\begin{equation}\label{eq:CauchyTr}
   \int_\CC\frac{|p_n(w)|^2\, \ee^{-NV(w)}}{z-w}  \dd A(w)={\cal O}(\ee^{-N \epsilon'}).
\end{equation} 
\end{thm}
\vspace{0.2cm}
\begin{cor} When $Q(z)=t z^2/2$, the Cauchy transform at \eqref{eq:CauchyTr} vanishes exponentially fast in $N$  when $z\in{\rm Int} \,K$.
\end{cor}  
\vspace{0.2cm}

Though it follows directly from Theorem 7.7.3 in \cite{ameur2011} that the above Cauchy transform vanishes in the leading order, it is new (at least to us) that the Cauchy transform vanishes at {\em all} orders in the $1/N$ expansion.    
The Cauchy transform in \eqref{eq:CauchyTr} is related to the Cauchy transform that appears in the Dbar problem \cite{its-takhtajan} through the following identity.
\begin{equation*}
\frac{1}{p_n(z)} \int_\CC\frac{|p_n(w)|^2\, \ee^{-NV(w)} }{z-w} \dd A (w)=\int_\CC\frac{\overline{p_n(w)}\, \ee^{-NV(w)}}{z-w} \dd A(w),  
\end{equation*}
which can be verified by the following equation that is immediate from the orthogonality of $p_n$'s.
\begin{equation*}
\int_\CC\frac{p_n(z)-p_n(w)}{z-w}\,\,\overline{p_n(w)}\, \ee^{-NV(w)} \dd A (w)=0.  
\end{equation*}

In the remainder of this paper, we will present the proofs of the four theorems.  In Section \ref{sec:def} we list the notations and a few useful (known) facts.  In Section \ref{sec:WKB} we give the asymptotic behavior of the orthogonal polynomial $p_n(z)$ using the properties of the Hermite polynomials. In Section \ref{sec:Bulk} we prove Theorem \ref{thm2}.   In Section \ref{sec_density} we prove Theorem \ref{thm1} and Theorem \ref{thm:1b}. In Section \ref{sec:vanishing-cauchy} we prove Theorem \ref{thm:Cauchy}.  In Section \ref{sec_discussion} we discuss several issues including the expected number of particles outside the ellipse.

\section{Useful definitions and facts}\label{sec:def}
The foci of the ellipse $\partial K$ lie at $\pm F_0$ where 
\begin{equation*} 
F_0=\sqrt{\frac{4\,t\,T}{1-t^2}}.
\end{equation*}
For $0<t<1$ we define the conformal map ($\DD$ is the unit disk)
\begin{equation}\label{eq_phi}
\phi:\CC\setminus \clos \DD \to \CC\setminus K
\qquad\text{by}\qquad 
\phi(u):=\frac{F_0}{2\sqrt{t}}\bigg(u+\frac{t}{u}\bigg).
\end{equation}
One can check that
$\phi(\partial \DD)=\partial K$.

The inverse function gives the conformal map 
\begin{equation}\label{eq_psi}
\psi: \CC\setminus K \to  \CC\setminus \clos \DD \qquad\text{by}\qquad\psi(z)=\frac{\sqrt{t} z}{F_0}\Big(1+\sqrt{1-F_0^2/z^2}\Big).
\end{equation}
The branch of the square root is taken such that
$\psi(z)=2\sqrt{t}\, z/F_0+{\cal O}(1)$ as $z$ goes to $\infty$.
It will be useful to denote the logarithmic capacity of $K$ by
\begin{equation}\label{eq:cap}
\logcap(K) = \lim_{z\to\infty}\frac{z}{\psi(z)}= \frac{F_0}{2\sqrt t}=\sqrt{\frac{T}{1-t^2}}.
\end{equation} 

The Joukowsky map $\phi$ is in fact a conformal map between the larger domains, i.e. 
$$\phi:\CC\setminus \clos B_{\sqrt t}(0) \to \CC\setminus[-F_0,F_0],$$
where $B_r(x)$ stands for the disk of radius $r$ centered at $x\in\CC$.
Therefore, $\psi$ is analytically continued to the conformal map
$$\psi:\CC\setminus[-F_0,F_0]\to \CC\setminus \clos B_{\sqrt t}(0).$$ Throughout the paper it is convenient to consider the expression
$$\sqrt{1-F_0^2/z^2}$$
to be analytic on $\CC\setminus[-F_0,F_0]$ such that the value at $\infty$ is one.

The {\em Schwarz function} $S$ for an analytic curve $\Gamma$ in the complex plane is defined to be the unique analytic function in a neighborhood of $\Gamma$ which fulfills $S(z)=\cc{z}$ on $\Gamma$.
More details about the Schwarz function can be found, for instance, in \cite{davis1974schwarz}.

For the ellipse, $\partial K$, the Schwarz function $S:\CC\setminus[-F_0,F_0]\to\CC$ is given by 
\begin{equation}
S(z)=t\,z + \frac{2 T \, z}{F_0^2}\left(1-\sqrt{1-F_0^2/z^2}\right).
\label{eq_schwarz_ellipse}
\end{equation}
This can be derived by the following identity
\begin{equation*}
S(z)=\overline{\phi(1/\overline{\psi(z)})},
\end{equation*}
which holds for a general simply-connected droplet.  It is simple to check that $S(z)$, as defined by the above equation, is analytic and satisfies $S(z)=\overline z$ on $z\in\partial K$.
\vspace{0.2cm}

\remark{For $t=0$ we have $\phi(u)=\sqrt{T}\,u$, $\psi(z)=z/\sqrt{T}$, $\logcap(K)=\sqrt{T}$ and $S(z)=T/z$.}

Let $p_j$ be the orthonormal polynomial of degree $j$ satisfying
$$ \int_\CC p_j(z)\overline{p_k(z)} \ee^{-NV(z)}\dd A(z)=\delta_{jk},\quad \lim_{z\to\infty}\frac{p_j(z)}{z^j}>0. $$
The following statement is known \cite{eijndhoven} (see also \cite{rr1312.0068}).

\vspace{0.2cm}
\begin{thm}\label{prop:Hermite}
Let $0<t<1$. The orthonormal polynomials $\{p_j\}_{j=0,1,2,\cdots}$ are given by
\begin{align}
p_j(z)=\bigg(\frac{t}{2}\bigg)^{j/2} (1-t^2)^{1/4}\sqrt{\frac{N}{\pi j!}} H_j\bigg(\sqrt\frac{N(1-t^2)}{2t}\, z\bigg), \label{eq_onp} 
\end{align}
where $H_j$ is the Hermite polynomial of degree $j$ given by the generating function
\begin{equation*}
\ee^{2x z- z^2}=\sum_{j=0}^\infty H_j(x)\frac{z^j}{j!}.
\end{equation*}
The leading coefficient of $p_j$ is given explicitly by
\begin{equation}
\gamma_j:=\lim_{z\to\infty} \frac{p_j(z)}{z^j}=
\frac{ N^{1/4}\sqrt{N(1-t^2)}^{j+1/2}}{\sqrt{\pi j!}}. 
 \label{eq_gamman}
\end{equation}
\end{thm}

We note that $\gamma_j$ and $p_j$ (and many other related quantities) also depend on the parameter $N$.  
\vspace{0.2cm}

\remark{If $t=0$ the orthonormal polynomials are simply the monomials $\{\gamma_j z^j\}$ where $\gamma_j=\frac{ N^{(j+1)/2 }}{\sqrt{\pi j!}}$. }

\vspace{0.3cm}
\begin{lemma}\label{lem:rho}
Given a fixed $r\in\ZZ$, as $n=NT$ goes to $\infty$ for a fixed $T$, the following expansion holds.
\begin{equation*}
\gamma_{n+r}	
=\left(\frac{N}{2\pi^3}\right)^{1/4} \sqrt{\frac{e^{-N\ell}}{(\logcap(K))^{2r+1}}} \bigg(1-\frac{1+6r+6r^2}{24 \,T N} +{\cal O}\bigg(\frac{1}{N^2}\bigg) \bigg),
\end{equation*}
where $\logcap(K)$ is defined at \eqref{eq:cap}, and the Robin constant $\ell$ is defined by
\begin{equation}\label{eq:ell}
	\ell := T\big(2\log \logcap(K)-1\big).
\end{equation}
\end{lemma}

One can prove this lemma by direct manipulation of \eqref{eq_gamman} using Stirling's formula for the factorial.    The leading asymptotic expression in the lemma also appeared in \cite{balogh2012strong} which is believed to be universal.  

Let us define the {\em pre-kernel} ${\cal K}_n$ by
\begin{align}\label{eq:pre}
\begin{split}
{\cal K}_n(z,w)&= K_n(z,w) \exp\bigg(\frac{N}{2}|z-w|^2\bigg) \exp\bigg(iN \,{\rm Im}\bigg(-z\overline w + \frac{t}{2} z^2 + \frac{t}{2} \overline w^2\bigg)\bigg)
\\ &=\sum_{j=0}^{n-1} p_j(z)\overline{p_j(w)} \exp\bigg(-N\bigg( z\overline w -\frac{t}{2}z^2-\frac{t}{2}\overline w^2\bigg)\bigg).
\end{split}
\end{align}
\begin{prop}\label{prop_CD}
Let $n\in\NN$ and $w,z\in\CC$. Then the pre-kernel ${\cal K}_n$ fulfills the following Christoffel-Darboux type identity.
\begin{align}
\frac{1}{N} \frac{\partial {\cal K}_n}{\partial z} (z,w) &= 
\frac{\sqrt{T}}{\sqrt{1-t^2}}\big(t\, p_{n}(z)\overline{p_{n-1}(w)}- p_{n-1}(z)\overline{p_n(w)}\big)\exp\bigg(-N\bigg( z\overline w -\frac{t}{2}z^2-\frac{t}{2}\overline w^2\bigg)\bigg).\label{eq_CD}
\end{align}
\end{prop}
This identity is known \cite{harnad} 
in more general cases in the context of the two matrix model and bi-orthogonal polynomials.  Below we reproduce the proof tailored to our case (See also \cite{rr1312.0068}).

\begin{proof}
We set $p_{-1}(z)=0$.  From Theorem \ref{prop:Hermite}  we have (see, for instance, \cite{rr1312.0068,elbau}), for $k\in \NN \cup\{0\}$,
\begin{align}\label{eq:rec}
\begin{split}
z \,p_{k}(z)&=r_{k+1} p_{k+1}(z)+t \, r_{k}  p_{k-1}(z), \\
\frac{1}{N}\partial_z \widehat p_{k}(z)&
=t\, r_{k+1} \widehat p_{k+1}(z)+ r_k \widehat p_{k-1}(z), 
\end{split}
\end{align}
where we set $r_{k}=\sqrt{\frac{k}{ N ( 1-t^2 ) }}$ and denote $\widehat p_k(z)=p_{k}(z)\,\ee^{Ntz^2/2}$. 
The second relation can be derived from 
$$p'_{k}(z)=\sqrt{N k(1-t^2)}\, p_{k-1}(z).$$
Using the definition of the pre-kernel \eqref{eq:pre} and the recurrence relations \eqref{eq:rec}, we get
\begin{align*}
\frac{\ee^{N z \overline w}}{N}\partial_z\bigg({\cal K}_n(z,w)\bigg)&=\frac{\ee^{N z \overline w}}{N}\partial_z\left(\sum_{j=0}^{n-1} \widehat p_j(z)\overline{\widehat p_j(w)} \ee^{-N z \overline w}\right)
\\&\!\!\!\!\!\!\!\!\!\!\!\!\!\!\!\!\!\!\!\!\!=
\sum_{j=0}^{n-1} \bigg(\big(t \,r_{j+1} \widehat p_{j+1}(z)+ r_j \widehat p_{j-1}(z)\big)\overline{\widehat p_j(w)}- \widehat p_j(z) \big(r_{j+1} \overline{\widehat p_{j+1}(w)}+t \, r_{j} \overline {\widehat p_{j-1}(w)}\big)\bigg)
\\&\!\!\!\!\!\!\!\!\!\!\!\!\!\!\!\!\!\!\!\!\!=
\sum_{j=0}^{n-1} \big(t \,r_{j+1} \widehat p_{j+1}(z)+ r_j \widehat p_{j-1}(z)\big)\overline{\widehat p_j(w)}- \sum_{j=1}^{n}\widehat p_{j-1}(z) r_{j} \overline{\widehat p_{j}(w)}-\sum_{j=0}^{n-2}\widehat p_{j+1}(z)\, t \, r_{j+1} \overline {\widehat p_{j}(w)}
\\&\!\!\!\!\!\!\!\!\!\!\!\!\!\!\!\!\!\!\!\!\!= 
t\, r_n\widehat p_{n}(z)\overline{\widehat p_{n-1}(w)}-r_n \widehat p_{n-1}(z)\overline{\widehat p_n(w)}.
\end{align*}
This proves the proposition because $r_n=\sqrt{T/(1-t^2)}$. 
\end{proof}

Using $K_n$ or ${\cal K}_n$, the average density $\rho_n$ in \eqref{eq:first-density} is given by (note that it is normalized to $T$)
\begin{equation}\label{eq:rho}
\rho_n(z)=\frac{1}{N} K_n(z,z)=\frac{1}{N} {\cal K}_n(z,z).
\end{equation}

\begin{cor}\label{cor_CD}
Let $n$ be a positive integer. The density $\rho_n$ satisfies
\begin{align}
 \frac{\partial\rho_n(z)}{\partial z} &=\bigl( t\,  p_n(z) \overline{p_{n-1}(z)} - p_{n-1}(z) \overline{p_n(z)} \bigr)\frac{\sqrt T\ee^{-N V(z)}}{\sqrt{1-t^2}}.\label{eq_CD2}
\end{align}
\end{cor}

\section{Asymptotic Expansion of Orthogonal Polynomials}\label{sec:WKB}
The orthonormal polynomials, $\{p_k\}$, are rescaled Hermite polynomials, $\{H_k\}$, where the argument of $H_k$ is multiplied by the constant $\sqrt{ N(1-t^2)/(2t)}$, see Theorem \ref{prop:Hermite}.  
The leading asymptotics for the Hermite polynomial has been proven by Plancherel and Rotach in \cite{plancherel-rotach}. See also \cite{deift,dkmvz1999,dkmvz2001,wong}.

We will need the asymptotics of $p_n$ to get the asymptotic behavior of the kernel and the density $\rho_n$ using Proposition \ref{prop_CD} and Corollary \ref{cor_CD}.  
Especially, we will use the (subleading) behaviors of $p_n$ to get the (subleading) behaviors of the kernel and the density.

The following lemma is shown in Appendix \ref{app:WKB}.
An alternative, self-contained method to calculate the asymptotics of classical orthogonal polynomials can be found in \cite{neuschel}.

\vspace{0.2cm}
\begin{lemma}\label{prop-WKB} Let $n=NT$ and $r\in\ZZ$.  As $N$ goes to infinity, the following asymptotic expansion holds uniformly over any compact subset of $\CC\setminus[-F_0,F_0]$.
\begin{equation}\label{eq_on_asymptotics}
p_{n+r}(z)=\left(\frac{N}{2\pi^3}\right)^{1/4}\psi(z)^r \sqrt{\psi'(z)} \,\exp\bigg[ n\, g(z) -\frac{N\ell}{2} + \frac{1}{N} h_{r}(z)+ {\cal O}\left(\frac{1}{N^2}\right)\bigg],
\end{equation}
where $\psi$ is given in \eqref{eq_psi}, $\ell$ in \eqref{eq:ell},
\begin{align}\label{eq:g}
g(z)& = \log z +\log \left(\frac{1}{2}+\frac{1}{2}\sqrt{1-F_0^2/z^2}\right)+\left(1+\sqrt{1-F_0^2/z^2}\right)^{-1}-\frac{1}{2},
\\ \label{eq:hr}
\text{and}\quad
h_r(z)&=\frac{1}{T}\bigg(\frac{F_0^2(1+2r)}{8(z^2-F_0^2)}-\frac{1+6 r+6r^2}{24\sqrt{1-F_0^2/z^2}}-\frac{5 F_0^2}{48(z^2-F_0^2)\sqrt{1-F_0^2/z^2}}\bigg).
\end{align}
\end{lemma}
Though the function $g$ \eqref{eq:g} has the logarithmic branch cut due to the term $\log z$, the expression $\exp(n g(z))$ is analytic in $\CC\setminus[-F_0,F_0]$.
The expression of $h_r$ \eqref{eq:hr} is analytic in $\CC\setminus[-F_0,F_0]$.
\vspace{0.2cm}
\begin{remark}
If $t=0$ we simply have $g(z)=\log z$ and $ h_r(z)=-\frac{1+6 r+6r^2}{24 T}. $
\end{remark}

The leading term of the asymptotic formula \eqref{eq_on_asymptotics} for the orthogonal polynomial has been conjectured in \cite{agam2002viscous} and appeared in \cite{balogh2012strong} for different $K$.  For general $K$, $g$ is given by the complex logarithmic potential generated by the uniform measure on $K$.

We will also need the following, rather crude, asymptotic behavior of $p_n$.

\vspace{0.3cm}
\begin{prop}\label{prop:crude} Let $n=NT$ grow to $\infty$ for a fixed $T$. The following estimate holds uniformly over the whole complex plane.
\begin{equation*}
\left|p_{n}(z)\,\ee^{-N(T g(z)-\ell/2)}\right|={\cal O}\big(N^{5/12}\big),
\end{equation*}
The same bound holds for $p_{n-1}$. \end{prop}

\begin{proof} One can use the uniform asymptotics stated at Theorem 1.34, Theorem 1.49 and Theorem 1.60 in \cite{dkmvz2001-v2}.  These theorems are written for monic polynomials (denoted by $\pi_j$'s in the cited paper).  Therefore we need the behavior of the normalization constant $\gamma_j$ in \eqref{eq_gamman} (which is different from $\gamma_j$ in the cited paper; our normalization is with respect to the area integral) which is given in Lemma \ref{lem:rho}.

The error bound ${\cal O}(N^{5/12})$ comes from the behavior of $p_n$ near the focal points, see the equation (1.58) in \cite{dkmvz2001-v2}, and from the behavior of $\gamma_n$ in Lemma \ref{lem:rho}. The former contributes by $N^{1/6}$ and the latter contributes by $N^{1/4}$.     
\end{proof}

\vspace{0.2cm}

\begin{lemma}\label{lem:zero} As $N$ goes to $\infty$ while $n=NT$, there exists $\epsilon>0$ such that
$$\rho_n(0)=\frac{1}{\pi}+\mathcal{O}(\ee^{-N \epsilon}).$$
\end{lemma}
\begin{proof}
At the origin the Hermite polynomials give (see for example \cite{hille}). 
\begin{align*}
H_m(0)&=\left\{\begin{array}{ll} 0 & \text{if $m$ is odd,}\\ (-1)^{m/2} m!/\big(\frac{1}{2}m\big)! & \text{if $m$ is even.} \end{array}\right. \label{eq_Hm0}
\end{align*}
Using this in \eqref{eq_onp} and \eqref{eq:rho}, the density at the origin is given by
\begin{align*}
\rho_n(0)&=\frac{1}{N}\sum_{j=0}^{n-1}|p_j(0)|^2
=\frac{\sqrt{1- t^2}}{\pi}  \sum_{l=0}^{\lfloor\frac{n-1}{2}\rfloor} \left(\frac{t}{2}\right)^{2l} \frac{ (2l)!}{ (l!)^2 }.
\end{align*}
Note that the power series in the right-hand side is a truncated Taylor series from  
\begin{equation}\label{eq_taylorseries}
\frac{1}{\sqrt{1-t^2}}
=\sum_{l=0}^{\infty} \left(\frac{t}{2}\right)^{2l} \frac{ (2l)!}{ (l!)^2 }.
\end{equation}
Choosing $\epsilon>0$ such that $x=\ee^{\epsilon/T}t<1$, we have
\begin{align*}
\ee^{N\epsilon} \left| \rho_n(0) - \frac{1}{\pi}\right|&= \ee^{n\epsilon/T}\frac{\sqrt{1-t^2}}{\pi}\Bigg|\sum^\infty_{l=\lfloor \frac{n-1}{2}\rfloor+1} \left(\frac{t}{2}\right)^{2l} \frac{(2l)!}{(l!)^2 }\Bigg| 
\\&\le \frac{\sqrt{1-t^2}}{\pi}\Bigg|\sum^\infty_{l=\lfloor \frac{n-1}{2}\rfloor+1} \left(\frac{x}{2}\right)^{2l} \frac{(2l)!}{(l!)^2 }\Bigg|< \frac{\sqrt{1-t^2}}{\sqrt{1-x^2}} .
\end{align*}
where, in the penultimate inequality, we used the index $l$ in the summation is $\geq n/2$ and therefore $\ee^{n\epsilon/T} \le \ee^{2l\epsilon/T}$. 
\end{proof}

\section{Kernel in the bulk}\label{sec:Bulk}

We define $\Omega:\CC\to\RR$ by
\begin{equation*}
	\Omega(z) = V(z) - 2 T {\rm Re}\,g(z) +\ell.
\end{equation*}
In fact, the value of $\ell$ \eqref{eq:ell} has been chosen such that $\Omega\equiv 0$ on $\partial K$.
One can show that, for any $z_0\in\partial K$,
\begin{equation}\label{Omega1}
	\Omega(z) 
	 = |z|^2-|z_0|^2-2{\rm Re}\,\int_{z_0}^z S(\zeta)\dd\zeta, 
\end{equation}
by taking the derivative and using the explicit forms of $g$ \eqref{eq:g} and $S$ \eqref{eq_schwarz_ellipse}.  (The independence of $\Omega$ on $z_0$ can be seen from $S(z_0)=\overline{z}_0$ which is obtained by taking the derivative of $\Omega$.) Since $S(z)$ has a purely imaginary jump on $[-F_0,F_0]$, $\Omega$ can be defined everywhere, with no discontinuity on $[-F_0,F_0]$.  The following lemma addresses the global minima of $\Omega$ (see \cite{rr1312.0068} for another proof).

\vspace{0.3cm}
\begin{lemma}\label{prop_Omega_minimum}
For $z\in\CC$ we have $\Omega(z)\geq 0$ and the equality holds only when $z\in\partial K$. 
\end{lemma}

\begin{proof} 
Let us define $\widetilde\Omega:\CC\to \RR$ by
\begin{equation*}
\widetilde\Omega(u):=	\Omega(\phi(u)) = |\phi(u)|^2-|\phi(1)|^2-2 {\rm Re}\int_1^u\phi(1/v)\,\dd\phi(v).
\end{equation*}
Since $\phi:\CC\setminus (\clos B_{\sqrt{t}}(0))\to \CC\setminus[-F_0,F_0]$ is conformal, it is enough to show that
\begin{equation}\label{eq:toprove}
\{u\in \CC\setminus B_{\sqrt{t}}(0)~|~\Omega(u)\leq 0\} =\partial\DD.
\end{equation}
Comparing the functions $g$ \eqref{eq:g} and $\psi$ \eqref{eq_psi} we see that 
\begin{align*}
g(\phi(u))=\bigg[\frac{t}{2(\psi(z))^2}-\log\left(\frac{2\sqrt{t}}{\psi(z) F_0} \right)\bigg]_{z=\phi(u)}= \frac{t}{2u^2}-\log\left(\frac{2\sqrt{t}}{u F_0} \right). 
\end{align*}
A straightforward calculation gives
\begin{align}
\widetilde\Omega(R\ee^{i\theta})
&=\frac{T (R^{-2}-1)}{1-t^2}\left(t^2-R^2+ (R^2-1) t \cos(2\theta) \right) -2T\log R.
\end{align}
The partial derivatives in $\theta$ and $R$ give
\begin{align}\label{eq_Omega_theta}
\frac{\partial\widetilde\Omega(R\ee^{i\theta})}{\partial \theta}&=\frac{2T}{1-t^2}(R-R^{-1})^2 t \sin(2\theta),
\\
\frac{\partial\widetilde\Omega(R\ee^{i\theta})}{\partial R}&=\frac{2T(R^{-1}-R^{-3})}{1-t^2}\left(R^2+t^2- (1+R^2) t \cos(2\theta)\right). \label{eq_Omega_R}
\end{align}
The first expression vanishes only when either $R=1$ or $\sin(2\theta)=0$.  When $R=1$, the expression \eqref{eq_Omega_R} vanishes.    One can also check that $\widetilde\Omega\equiv 0$ when $R=1$.
When $\sin(2\theta)=0$, the expression \eqref{eq_Omega_R} 
vanishes only when $R=\sqrt t$ and $R=1$.

We will show that $\widetilde\Omega>0$ for $R=\sqrt{t}$. There we find
\begin{align*}
\widetilde\Omega(\sqrt t\ee^{i\theta})=-\frac{2T(1-t)(\cos \theta)^2}{1+t}-T \log t\ge -\frac{2T(1-t)}{1+t}-T \log t, \qquad  0<t<1.
\end{align*}
The right-hand side goes to zero as $t$ goes to 1, and it is strictly positive for $0<t<1$ because its derivative is strictly negative, i.e.
\begin{align*}
-\frac{T(1-t)^2}{t(1+t)^2}<0, \qquad 0<t<1.
\end{align*}
Since $\widetilde\Omega$ goes to $+\infty$ as $R$ goes to $\infty$, we proved \eqref{eq:toprove}. 
\end{proof}
\begin{lemma}\label{lem:density}
Let $V_K$ be a compact subset of ${\rm Int}\,K$ and $U_K$ a neighborhood of $K$. Consider the limit where $N$ goes to $\infty$ while $n=NT$. Then there exists $\epsilon>0$ such that the following asymptotic behavior holds uniformly over all $z$ as specified below.
\begin{empheq}[left={\rho_n(z)=\empheqlbrace}]{align}
&\frac{1}{\pi} + {\cal O}(\ee^{-N\epsilon}),& z\in V_K,\label{eq:rhoout} \vspace{0.2cm}\\
&{\cal O}(\ee^{-N\epsilon}),& z\in U_K.\label{eq:rhoin}
\end{empheq}
\end{lemma}

\begin{proof} From Corollary \ref{cor_CD} and Proposition \ref{prop:crude}, there exists a constant $C>0$ such that
\begin{align}\label{eq:rho-crude-deriv}
 \bigg|\frac{\partial\rho_n(\xi)}{\partial \xi}\bigg| &\leq\Big(\big|\overline{p_n(\xi)} p_{n-1}(\xi)\big| + t\big|\overline{p_{n-1}(\xi)} p_n(\xi)\big|\Big)\frac{\sqrt T\ee^{-N{\rm Re} V(\xi)}}{\sqrt{1-t^2}}
\leq
C N^{5/6}\ee^{-N\Omega(\xi)}.
\end{align}
Since $\Omega$ is strictly positive over a compact subset of ${\rm Int}\,K$ by Lemma \ref{prop_Omega_minimum}, the right-hand side is uniformly bounded by ${\cal O}(\ee^{-N\epsilon})$ where $\epsilon>0$ can be chosen to be smaller than the minimal value of $\Omega$ over the compact subset.  To prove \eqref{eq:rhoout} we only need to integrate $\partial_\xi\rho_n$ from the origin to the point $z$ via straight path and use the behavior of $\rho_n(0)$ in Lemma \ref{lem:zero}. 

Using the same idea, the estimate \eqref{eq:rho-crude-deriv} also holds over {\em the complement of a neighborhood $U_K$ of $K$}.  Again, an integration starting from a fixed point, say $z_1$, in the complement of $U_K$ gives that 
\begin{equation*}
  	\rho_n(z) = c(N) + {\cal O}(\ee^{-N\epsilon})
  \end{equation*}  
where the error term is uniform over the complement of $U_K$ and $c(N)$ is given by $\rho_n(z_1)$.   We note that, even though the complement of $U_K$ is unbounded, the integration converges because $\Omega(\xi)$ grows quadratically as $|\xi|\to\infty$.    Then we must have $c(N)={\cal O}(\ee^{-N\epsilon})$ because, otherwise, $\int_{\CC\setminus U_K}\rho_n\dd A$ diverges.  This proves \eqref{eq:rhoin}.  
\end{proof}
 
Now we prove Theorem \ref{thm2}.

For any $z$ and $w$ the pre-kernel is given by
\begin{align}\label{eq:pre-Kernel-int}
\frac{1}{N}{\cal K}_n(z,w)=\frac{1}{N}{\cal K}_n(w,w)+\frac{1}{N}\int_w^z \frac{\partial {\cal K}_n}{\partial \xi}(\xi,w) \dd\xi.
\end{align}
The upper bound of the integrand comes from Proposition \ref{prop_CD} and Proposition \ref{prop:crude}:
\begin{align}\nn
\frac{1}{N} \left|\frac{\partial {\cal K}_n}{\partial \xi} (\xi,w)\right|
&\leq
\frac{\sqrt{T}}{\sqrt{1-t^2}}\Big(t\,\big| p_{n}(\xi)\overline{p_{n-1}(w)}\big|+ \big|p_{n-1}(\xi)\overline{p_n(w)}\big|\Big)\exp\bigg[\!\!-N{\rm Re}\Big( \xi\overline w -\frac{t}{2}\xi^2-\frac{t}{2}\overline w^2\Big)\bigg]
\\\label{eq:exp-small}&\leq  
C\frac{\sqrt{T}}{\sqrt{1-t^2}} N^{5/6}
\exp\bigg[\!\!-N{\rm Re}\Big( \xi\overline w -\frac{t}{2}\xi^2-\frac{t}{2}\overline w^2-T g(\xi)-Tg(w)+\ell\Big)\bigg]
\end{align}
for some $C>0$ and for all $N$ sufficiently large.
The exponent can be further simplified by 
\begin{equation}\label{eq:32}
	{\rm Re}\Big( \xi\overline w -\frac{t}{2}\xi^2-\frac{t}{2}\overline w^2-T g(\xi)-Tg(w)+\ell\Big)
	=-\frac{1}{2}|\xi-w|^2 +\frac{1}{2}\Omega(\xi) +\frac{1}{2}\Omega(w).
\end{equation}
By Proposition \ref{prop_Omega_minimum} we can define $\epsilon$ such that
\begin{equation*}
0<	2\epsilon<\min_{w\in V_K}\Omega(w).
\end{equation*}
When $\xi=w\in V_K$, the right-hand side of \eqref{eq:32} is greater than $2\epsilon$.  
Then there exists sufficiently small $\delta$ such that when $|\xi-w|<\delta$ the expression \eqref{eq:32} remains bigger than $\epsilon$ for any $w\in V_K$.  This implies that the upper bound \eqref{eq:exp-small} is given by ${\cal O}(\ee^{-N\epsilon})$ uniformly over $w\in V_K$.  The first case of Theorem \ref{thm2} follows from \eqref{eq:pre-Kernel-int} and Lemma \ref{lem:density}.

For the second case of $z\notin B_\delta(w)$, we will use
\begin{equation}\label{eq:error333}
	\frac{1}{N} \left|\frac{\partial {\cal K}_n}{\partial \xi} (\xi,w)\right| \exp\left(-\frac{N}{2}|\xi-w|^2\right)={\cal O}(\ee^{-N\epsilon}\ee^{-N\Omega(\xi)/2}), \qquad\xi\in\CC,
\end{equation}
which can be noticed from \eqref{eq:32}.  
Using this estimate, we get, by integration of $\partial_\xi {\cal K}_n(\xi,w)$ along the straight path from $w$ to $z$,
\begin{align*}
\frac{1}{N}{\cal K}_n(z,w) \ee^{-\frac{N}{2}|z-w|^2}&=\frac{1}{N}{\cal K}_n(w,w) \ee^{-\frac{N}{2}|z-w|^2}+\mathcal{O}(\ee^{-N\epsilon}), &&\text{$z\notin B_\delta(w)$}.
\end{align*}
We can use Lemma \ref{lem:density} to evaluate $\frac{1}{N}{\cal K}_n(w,w)$ for $w\in V_K$. Going back to the original kernel $K_n(z,w)$ by \eqref{eq:pre} proves the second case of Theorem \ref{thm2}. Again, as in the  proof of Lemma \ref{lem:density}, even though $z$ can be in the unbounded $\CC$, the error bound above is uniform since the error in \eqref{eq:error333} has $\Omega(\xi)$ in the exponent which grows quadratically for large $|\xi|$.

The last case when $w\notin U_K$ follows from choosing $\epsilon$ such that
\begin{equation*}
	0< 2\epsilon<\min_{w\notin U_K} \Omega(w).
\end{equation*}
The rest of the proof remains the same as in the second case.

\section{Density near the boundary}\label{sec_density}

We will prove Theorem \ref{thm1} using the Christoffel-Darboux formula in Proposition \ref{prop_CD} and the asymptotics of the orthonormal polynomials $p_n$ obtained in Section \ref{sec:WKB}.

For any complex number ${\bf v}=v_x+i v_y$ let us denote 
$$\partial_{\bf v} f(z)=({\bf v} \partial+ \cc{{\bf v}} \cc{\partial}) f(z)$$ where $\partial=\frac{\partial}{\partial z}=\frac{1}{2}\left(\frac{\partial}{\partial x}-i\frac{\partial}{\partial y}\right)$ and $\cc{\partial}=\frac{\partial}{\partial \cc{z}}=\frac{1}{2}\left(\frac{\partial}{\partial x}+i\frac{\partial}{\partial y}\right)$.

Let $\gamma:\RR\to \partial K$ be the counter-clockwise parametrization given by
\begin{equation*}
  \gamma(\theta)=\phi(e^{i\theta}).
\end{equation*}
The (outward) normal vector, ${\bf n}$, and the tangential vector, ${\bf t}$, at $\phi(\ee^{i\theta})$ are given by
\begin{equation*}
	{\bf t}=i {\bf n}=\frac{\gamma'(\theta)}{|\gamma'(\theta)|}.
\end{equation*}
We will denote the derivatives in normal direction and tangential direction by $\partial_{\bf n}$ and $\partial_{\bf t}$ respectively.

We denote the {\em derivative with respect to the arclength} by $\partial_s$, i.e.
$\partial_s = |\phi'|^{-1}\partial_\theta$. Obviously $\partial_s=\partial_{\bf t}$.  Generally, $\partial^k_s\neq \partial^k_{\bf t}$ for $k>1$.

\vspace{0.2cm}
\begin{prop}\label{prop_dv_density}
Let $z\in \CC\setminus [-F_0,F_0]$ and ${\bf v}\in \CC$. As $n=NT$ goes to $\infty$, we have the asymptotic expansion
\begin{align}\label{eq:23} 
\partial_{\bf v} \rho_n(z)&=\left(\frac{2N}{\pi^3}\right)^{1/2}  \ee^{-N \Omega(z)}   {\rm Re}\left[ \frac{\sqrt{T}(\cc{{\bf v}} t-{\bf v})}{\sqrt{1-t^2}} \frac{|\psi'(z) |}{ \psi(z) } \left(1+\frac{\cc{h_0(z)}+h_{-1}(z)}{N}+ \mathcal{O}\bigg(\frac{1}{N^2}\bigg)\right) \right].   
\end{align}
The error bound holds uniformly on any compact subset of $\CC\setminus [-F_0,F_0]$.
\end{prop}

\begin{proof}
From the Christoffel-Darboux formula \eqref{eq_CD2} we get
\begin{align}\label{eq:35}
\partial_{\bf v} \rho_n(z)&=2\,{\rm Re}\Big( (\cc{{\bf v}} t-{\bf v}) \overline{p_n(z)}\, p_{n-1}(z) \Big) \frac{\sqrt{T} \ee^{ -N V(z)}}{\sqrt{1-t^2} }.   
\end{align}
From Lemma \ref{prop-WKB} we have
\begin{align*}
p_n(z)&=\left(\frac{N}{2\pi^3}\right)^{1/4} \sqrt{\psi'(z)}\, \ee^{N(T g(z)-\ell/2)} \left(1+\frac{h_{0}(z)}{N} +\mathcal{O}\bigg(\frac{1}{N^2}\bigg)\right),
\intertext{and}
p_{n-1}(z)&=\left(\frac{N}{2\pi^3}\right)^{1/4} \frac{\sqrt{\psi'(z)}}{ \psi(z) } \,\ee^{N(T g(z)-\ell/2)} \left(1+\frac{h_{-1}(z)}{N} +\mathcal{O}\bigg(\frac{1}{N^2}\bigg)\right),
\end{align*}
where the error terms hold uniformly on any compact subset of $\CC\setminus [-F_0,F_0]$. Putting these into \eqref{eq:35} proves the proposition.
\end{proof}

\subsection{Expansion of $\Omega$ near the boundary}

Recall that $\partial_s$ is the derivative with respect to the arclength and that ${\bf n}$ is the (outward) normal vector at the boundary $\partial K$. The (signed) curvature of the curve $\partial K$ is given by
\begin{align}\label{eq_kappa}
\kappa=\frac{\dd}{\dd s}\arg {\bf n}= -i \partial_s \log {\bf n}=-\frac{i}{{\bf n}}\partial_s {\bf n}.
\end{align}

On the boundary we find the following relation between ${\bf n}$, $\kappa$ and $S$.

\vspace{0.2cm}
\begin{lemma}\label{lem_ntkappa}
The following identity holds for $z_0\in\partial K$.
\begin{equation}\label{eq:3eq}
  {\bf n}^2 S'(z_0)=-1,\quad {\bf n}^3 S''(z_0)=2\kappa,\quad {\bf n^4}S'''(z_0)=-6\kappa^2-2i\partial_{s}\kappa.
\end{equation}
\end{lemma}
Analogous relations can be found in \cite{davis1974schwarz}.
\begin{proof}
Let $\gamma:\theta\mapsto\phi(\ee^{i\theta})$ be a parametrization of $\partial K$.  Then we get the tangential vector ${\bf t}$  by 
$${\bf t}=\frac{\partial_\theta\gamma(\theta)}{|\gamma'(\theta)|}=\frac{i u \,\phi'(u)}{|\phi'(u)|}\bigg|_{u=e^{i\theta}}=i |\psi'(z_0)| \frac{\psi(z_0)}{\psi'(z_0)}, \quad z_0=\phi(\ee^{i\theta}). $$
The Schwarz function satisfies $S(\phi(u))=\overline{\phi(1/\overline u)}$ in a neighborhood of $\partial K$.   Taking the derivative in $u$ we get
 $$S'(\phi(u))\phi'(u)=-\frac{\cc{\phi'(1/\overline u)}}{u^{2}} $$ 
and, when $z_0\in \partial K$, we have
\begin{equation}\label{eq:Sn}
S'(z_0)=-\frac{\psi'(z_0)}{\cc{\psi'(z_0)} (\psi(z_0))^2} = -\frac{1}{{\bf n}(z_0)^2} 
\end{equation} 
using ${\bf n}=-i{\bf t}$.  This proves the first equation of \eqref{eq:3eq}.

Taking the $\partial_s$-derivative on both sides of \eqref{eq:Sn} we get
\begin{equation}
  i{\bf n}(z_0) S''(z_0)=2\frac{1}{{\bf n}(z_0)^3} i{\bf n}(z_0)\kappa
\end{equation}
where we used $\partial_s=i{\bf n}\partial - i\overline{\bf n}\overline\partial$  on the left-hand side, and \eqref{eq_kappa} on the right-hand side.  This proves the second equation of \eqref{eq:3eq}.

The last equation is obtained by taking $\partial_s$-derivative of the second equation in the lemma.  (When repeated, such process can yield the relations with higher order derivatives than are shown in the lemma.)
\begin{equation*}
  3{\bf n}^2 (i {\bf n}) \kappa S'' + i {\bf n}^4 S''' = 2\partial_s\kappa.
\end{equation*}
Using the second equation in the lemma to replace $S''$, one gets the last equation.
\end{proof}

{\rm Remark.} Lemma \ref{lem_ntkappa} and the following lemma are both true for more general droplets.  Note that in the proof we never use the fact that $K$ is an ellipse. 
\vspace{0.2cm}

\begin{lemma}\label{prop_omega}
Let $z_0\in \partial K$.   Let $X, Y\in \RR$ be defined by
$$ X {\bf n}(z_0)+i  Y {\bf n}(z_0)= z -z_0.$$
For any $\nu$ such that $0\leq \nu<1/6$, as $N$ grows to infinity the following asymptotic expansions holds
\begin{align}
\ee^{-N\Omega(z)}=\ee^{-2N  X^2}\bigg(1&+\frac{2\kappa N}{3} (X^3-3 XY^2 )-\frac{\kappa^2 N}{2} (X^4-6 X^2 Y^2+ Y^4) \nn \\ &+ \frac{2\partial_s \kappa N}{3}(X^3 Y-X Y^3 ) +\frac{2\kappa^2 N^2}{9} (X^3-3 XY^2 )^2 +\mathcal{O}(N^{3+9(\nu-1/2)})\bigg). \label{eq_exp_omega}
\end{align}
uniformly over $|z-z_0|<N^{\nu-1/2}$ and over $z_0\in\partial K$.
\end{lemma}

\begin{proof}
Using $\partial_{\bf n}={\bf n}\partial + \overline{\bf n}\overline\partial$ and $\partial_{\bf t}=i{\bf n}\partial -i \overline{\bf n}\overline\partial$ and the first equation in \eqref{eq:3eq}, one finds that
\begin{equation}\label{eq:k2}
  \partial_{\bf n} \Omega(z_0) =   \partial_{\bf t} \Omega(z_0)= \partial_{\bf n}\partial_{\bf t} \Omega(z_0)= 
  \partial_{\bf t}^2 \Omega(z_0) =0. 
\end{equation}
As an example, let us prove the last equality.
\begin{equation*}\begin{split}
 \partial_{\bf t}^2 \Omega(z_0)= (i{\bf n}\partial-i\overline{\bf n}\overline\partial)^2\Omega(z)\big|_{z=z_0}
 &= (i{\bf n}\partial-i\overline{\bf n}\overline\partial) \left(i{\bf n}(\overline z- S(z)-i\overline {\bf n}(z- \overline{S(z)}\right)\Big|_{z=z_0}
 \\
 &= \left({\bf n}^2 S'(z)+\overline{\bf n}^2\overline{S'(z)}+2\right)\Big|_{z=z_0} = 0,
 \end{split}
\end{equation*}
where, in the last equality, we used the first equation in Lemma \ref{lem_ntkappa}.  Similarly, we get
\begin{equation}\label{eq:k22}
  \partial^2_{\bf n} \Omega(z_0)
  = \left(-{\bf n}^2 S'(z)-\overline{\bf n}^2\overline{S'(z)}+2\right)\Big|_{z=z_0}=4,
\end{equation}
using the first equation in Lemma \ref{lem_ntkappa}.

In the derivatives of order three or more, the gaussian term $|z|^2$ in $\Omega(z)$ does not contribute.  So we have
\begin{equation*}
  \partial_{\bf v_1} \partial_{{\bf v}_2} \cdots \partial_{{\bf v}_k} \Omega(z)=-2 \partial_{\bf v_1} \partial_{{\bf v}_2} \cdots \partial_{{\bf v}_k} {\rm Re}\left(\int^z S(\zeta)\dd \zeta\right)= -2{\rm Re}\left({\bf v}_1{\bf v}_2 \cdots {\bf v}_k S^{(k-1)}(z)\right),
\end{equation*}
for $k\geq 3$.  Using the second and the third equations in Lemma \ref{lem_ntkappa}, one can obtain the following identities for the third order derivatives:
\begin{equation}\label{eq:k3}
  \partial_{\bf n}^3\Omega(z_0)=-\partial_{\bf n}\partial_{\bf t}^2\Omega(z_0)=-4\kappa, \quad \partial_{\bf t}^3\Omega(z_0)=-\partial_{\bf n}^2\partial_{\bf t}\Omega(z_0)=0,
\end{equation}
and for the fourth order derivatives:
\begin{equation}\label{eq:k4}
  \partial_{\bf n}^4\Omega(z_0)=\partial_{\bf t}^4\Omega(z_0)=-\partial_{\bf n}^2\partial_{\bf t}^2\Omega(z_0)=12\kappa^2,\quad \partial_{\bf n}^3\partial_{\bf t}\Omega(z_0)=-\partial_{\bf n}\partial_{\bf t}^3\Omega(z_0)=-4\partial_s\kappa. 
\end{equation}

Since both the real and the imaginary parts of $\Omega$ are real analytic, we have  
\begin{align}\label{eq_taylor}
\Omega(z)&=\sum_{k=0}^l \frac{(X\partial_{{\bf n}}+Y\partial_{{\bf t}} )^{k} \Omega(z_0)}{k!}  + {\cal O}(|z-z_0|^{l+1})
\end{align}
uniformly over $|z-z_0|<\delta$ and over $z_0\in\partial K$, for a sufficiently small $\delta$.
Putting \eqref{eq:k2}\eqref{eq:k22}\eqref{eq:k3}\eqref{eq:k4} into the Taylor expansion, we get
\begin{align}\label{eq_omega_approx}
\Omega(z)&=2 X^2-\frac{2\kappa}{3} (X^3-3 XY^2 )+\frac{\kappa^2}{2} (X^4-6 X^2 Y^2+ Y^4)- \frac{2\partial_s \kappa}{3}(X^3 Y-X Y^3 )+{\cal O}(|z-z_0|^5).\end{align}
Multiplying the above by $N$ and exponentiating, we get
\begin{align*}
\exp\big(-N\Omega(z)\big)=\ee^{-2 N X^2}\bigg(1&+ N\frac{2\kappa}{3} (X^3-3 XY^2 )+N^2\frac{2\kappa^2}{9} (X^3-3 XY^2 )^2
\\&-N\frac{\kappa^2}{2} (X^4-6 X^2 Y^2+ Y^4)+ N\frac{2\partial_s \kappa}{3}(X^3 Y-X Y^3 )
\\&\qquad\qquad\ \,  +{\cal O}\left(N|z-z_0|^5,N^3|z-z_0|^9,N^2|z-z_0|^7\right)\bigg),
\end{align*}
uniformly over $|z-z_0|<N^{-1/3}$ and over $z_0\in\partial K$, for a sufficiently large $N$ because the exponential function converges uniformly in a bounded region.  The lemma follows because, among the error terms, $N^3|z-z_0|^9$ contributes the most.
\end{proof}

\subsection{Proof of Theorem \ref{thm:1b}}

Theorem \ref{thm:1b} that we prove in this section is a stronger version of Lemma \ref{lem:density}. The theorem gives an estimate in a neighborhood of $\partial K$ that scales arbitrary close to the boundary as $N$ goes to infinity.

\begin{proof}
From \eqref{eq_omega_approx}
we have 
\begin{align}\label{eq_expansion2}
\Omega(z)&=2X^2+ {\cal O}(|z-z_0|^{3}),
\end{align}
uniformly over $|z-z_0|<\delta$ and over $z_0\in\partial K$, for a sufficiently small $\delta>0$.  

We denote the neighborhood of $\partial K$ of width $2\delta$ by
\begin{align*}
U_{\delta}&=\{z_0+ X {\bf n}(z_0)\ |\ z_0\in\partial K, |X| <\delta \}. 
\end{align*}
For any 
$z\in U_\delta\setminus U_{\tau,N}$ there exists $z_0\in\partial K$ such that $z=z_0+X {\bf n}$ with $N^{-1/2+\tau} \le |X|<\delta$.
From \eqref{eq_expansion2}, we have 
$ \Omega(z) \ge 2 N^{-1+2\tau} + {\cal O}(N^{-3/2+3\tau}) > N^{-1+2\tau}$ on $z\in U_\delta\setminus U_{\tau,N}$ for any sufficiently large $N$. 
This leads to
\begin{equation}\label{eq:UU}
 \ee^{-N\Omega(z)}<\ee^{-N^{2\tau}}~~ \text{ on $z\in U_\delta\setminus U_{\tau,N}$.}
\end{equation}
As in Lemma \ref{lem:density} we now use \eqref{eq:rho-crude-deriv} to get a uniform bound for the derivative of the density on $U_\delta\setminus U_{\tau,N}$. 
We get the density at the point $z=z_0+X{\bf n} \in U_\delta\setminus U_{\tau,N}$ with $z_0\in \partial K$ and $N^{-1/2+\tau} \le |X|<\delta $ by integrations from starting points $z_0-\delta {\bf n}\in {\rm Int}\, K$ or from $z_0+\delta {\bf n}\in {\rm Ext}\, K$.
When $N$ is large enough, from \eqref{eq:rho-crude-deriv} and \eqref{eq:UU}, the contribution from these integrations are bounded by
$$\left|\rho_n(z)-\rho_n(z_0\pm \delta{\bf n})\right| \leq  \delta \cdot C N^{5/6} \ee^{-N^{2\tau}} \quad \text{ when }z\in U_\delta\setminus U_{\tau,N}.$$ From Lemma \ref{lem:density} we know that at the starting points $z_0\mp\delta {\bf n} \in \partial U_\delta$ the density equals $1/\pi + {\cal O}(\ee^{-N\epsilon})$ on ${\rm Int}(K)\setminus U_\delta$ and ${\cal O}(\ee^{-N\epsilon})$ on ${\rm Ext}(K)\setminus U_\delta$ when $\epsilon>0$ is choosen small enough.
\end{proof}

\subsection{Proof of Theorem \ref{thm1}}

The proof of the following three lemmas are in Appendix \ref{app-proof3}.

\vspace{0.2cm}
\begin{lemma}\label{prop_psi} Let  $z_0\in \partial K$. 
  Let $X, Y\in \RR$ be defined by
$$ X {\bf n}(z_0)+i  Y {\bf n}(z_0)= z -z_0.$$
The following expansion holds as $z$ goes to $z_0$
\begin{align*}
\frac{|\psi'(z)|}{\psi(z)}&= \frac{|\psi'(z_0)|}{\psi(z_0)}\left(1+\Theta_{1,0}(z_0) X +\Theta_{0,1}(z_0) Y+ \frac{\Theta_{2,0}(z_0)}{2} X^2 + \Theta_{1,1}(z_0) X Y+ \frac{\Theta_{0,2}(z_0)}{2} Y^2 \right)
\\ &\qquad + \mathcal{O}\left(|z-z_0|^3\right),
\end{align*}
uniformly over $|z-z_0|<\delta$ and over $z_0\in\partial K$, for a sufficiently small $\delta$.  We used the definitions:
\begin{align*}
\Theta_{1,0}(z_0)&=-\kappa, \qquad\quad \Theta_{0,1}(z_0)= \frac{1}{3}\frac{\partial_s\kappa}{\kappa}-i|\psi'(z_0)|,\qquad\quad \Theta_{1,1}(z_0)= -\frac{5}{3} \partial_s \kappa + 2 i \kappa|\psi'(z_0)|, \\
\Theta_{2,0}(z_0)&=2\kappa^2+\frac{(\partial_s \kappa)^2}{3 \kappa^2}-\frac{\partial_s^2 \kappa}{3 \kappa}  + i \frac{\partial_s\kappa}{3 \kappa} |\psi'(z_0)|,    \\
\Theta_{0,2}(z_0)&=-\frac{2}{9}\frac{(\partial_s\kappa)^2}{\kappa^2}-\kappa^2+\frac{\partial_s^2 \kappa}{3\kappa}-|\psi'(z_0)|^2  -i |\psi'(z_0)| \frac{\partial_s\kappa}{\kappa}. 
\end{align*}
\end{lemma}

\begin{lemma}\label{lemma_hn_relations}
Let $z_0\in{\partial K}$. Then
\begin{align*}
{\rm Re}\left( h_{0}(z_0) +h_{-1}(z_0)\right)&=-\frac{1}{12} \kappa^2+\frac{1}{24} \frac{\partial_s^2 \kappa}{\kappa},\\
\frac{ {\rm Im}(h_{0}(z_0)-h_{-1}(z_0))}{|\psi'(z_0)|}&=\frac{\partial_s \kappa}{6 \kappa}.
\end{align*}
\end{lemma}

\begin{lemma}\label{lemma_Upsilon}
Let $z_0\in{\partial K}$. Then
\begin{align*}
& \frac{\sqrt{T}(\cc{{\bf n}} t-{\bf n})}{\sqrt{1-t^2}} \frac{|\psi'(z_0)|}{\psi(z_0)}=-1 + \frac{i \partial_s \kappa}{3|\psi'(z_0)| \kappa},  \\
&\frac{\sqrt{T}(\cc{{\bf n}} t+{\bf n})}{\sqrt{1-t^2}} \frac{|\psi'(z_0)|}{\psi(z_0)}= \frac{ \kappa  }{ |\psi'(z_0) |}.
\end{align*}
\end{lemma}

\vspace{0.2cm}

\begin{prop}\label{prop_derivative_density} 
Let $z_0\in \partial K$.  
Let us define $X, Y\in \RR$ by
$$ X {\bf n}(z_0)+i  Y {\bf n}(z_0)= z -z_0.$$ 
For any $0\leq \nu<1/6$ the following asymptotic expansions hold.
\begin{align}
\partial_{\bf n} \rho_n(z) &=\sqrt{\frac{N}{2\pi^3}}\ee^{-2 N X^2}  \bigg(-2-\kappa\left(\tfrac{4}{3}N X^3-2X\right)- \kappa^2 \left(\tfrac{4}{9}N^2 X^6-\tfrac{7}{3}N X^4+2X^2\right) \label{eq_derivative_n_density} 
\\\nn  &  \quad\qquad\qquad\qquad\qquad +\left(\frac{\partial_s^2 \kappa}{3\kappa} -  \frac{4(\partial_s\kappa)^2}{9\kappa^2}\right) X^2 
+\frac{1}{N}\left(\frac{\kappa^2}{6} +\frac{(\partial_s\kappa)^2}{9 \kappa^2}-\frac{\partial_s^2\kappa}{12\kappa}\right)+\mathcal{O}\left(N^{-3/2+9\nu} \right)  \bigg), 
\\
\partial_{\bf t} \rho_n(z)&=\sqrt{\frac{N}{2\pi^3}}\ee^{-2 N X^2}  \bigg( -2 \kappa Y  + \kappa^2 \left( 4N XY^3-\tfrac{4}{3} NX^3Y+4XY  \right) 
  +\partial_s\kappa \Big(\frac{X^2}{3}-Y^2 \Big) \nn \\
&\quad\qquad\qquad\qquad\qquad  -\frac{1}{N}\frac{\partial_s \kappa}{3|\psi'(z_0)|}
 +\mathcal{O}\left(N^{-3/2+7\nu}\right)  \bigg),  \label{eq_derivative_t_density}
\end{align}
where the former is evaluated for $Y=0$.
The error bounds are uniform over $\{z ~|~|z-z_0|<N^{-1/2+\nu}\}$ and over $z_0\in\partial K$.
\end{prop}

\begin{proof}
We use Lemma \ref{prop_psi} and Lemma \ref{lemma_Upsilon} to get
\begin{align}\nn
&	\quad	{\rm Re}\!\left[ \frac{\sqrt{T}(\cc{{\bf n}} t-{\bf n})}{\sqrt{1-t^2}} \frac{|\psi'(z) |}{ \psi(z) } \left(1+\frac{\cc{h_0(z)}+h_{-1}(z)}{N}+ \mathcal{O}\left(\frac{1}{N^2}\right)\right) \right]
\\ &	\qquad\qquad	=	 \bigg[\kappa X-1 
		-\left(\kappa^2 +\frac{2(\partial_s \kappa)^2}{9\kappa^2}
		-\frac{\partial_s^2\kappa}{6\kappa}\right) X^2
\\\nn & \qquad\qquad\qquad		  
- \frac{{\rm Re}\left(h_0(z_0)+h_{-1}(z_0)\right) }{N}  +\frac{\partial_s\kappa\,{\rm Im}\left(h_0(z_0)-h_{-1}(z_0)\right) }{3 \kappa |\psi'(z_0)| N  }  \bigg] + \mathcal{O}(N^{3\nu-3/2}).
\end{align}	
We use the above with Lemma \ref{prop_omega} in equation \eqref{eq:23} with ${\bf v}={\bf n}$ to get
\begin{align}
\partial_{\bf n}\rho_n(z)=&-\sqrt{\frac{2N}{\pi^3}}\ee^{-2 N X^2}  \Bigg(1+\kappa\left(\tfrac{2}{3}N X^3-X\right)+ \kappa^2 \left(X^2-\tfrac{7}{6}N X^4+\tfrac{2}{9}N^2 X^6\right) +\left(\frac{2(\partial_s\kappa)^2}{9\kappa^2} - \frac{\partial_s^2 \kappa}{6\kappa}\right) X^2 \nn \\
&\qquad\qquad +\frac{{\rm Re}(h_0(z_0)+h_{-1}(z_0)) }{N} +\frac{\partial_s\kappa}{3 |\psi'(z_0)| \kappa}\frac{{\rm Im}(h_{-1}(z_0)-h_{0}(z_0))}{N}+\mathcal{O}(N^{9\nu-3/2})  \Bigg).\label{eq:with-hn} 	
\end{align}
One can use Lemma \ref{lemma_hn_relations} to finally get \eqref{eq_derivative_n_density}.

Similar steps using Lemma \ref{prop_psi}, Lemma \ref{lemma_Upsilon} and Lemma \ref{prop_omega} in equation \eqref{eq:23} with ${\bf v}={\bf t}$ give 
\begin{align}
	\partial_{\bf t} \rho_n(z)&=\sqrt{\frac{N}{2\pi^3}}\ee^{-2 N X^2}  \Bigg( -2 \kappa Y  + \kappa^2 \left( 4N XY^3-\tfrac{4}{3} NX^3Y+4XY  \right) 
  +\partial_s\kappa \left(\frac{X^2}{3}-Y^2 \right) \nn \\
&\quad\qquad\qquad\qquad\qquad  +\frac{2\kappa}{|\psi'(z_0)|}\frac{{\rm Im}\left(h_{-1}(z_0)-h_{0}(z_0)\right)}{N} +\mathcal{O}\big(N^{7\nu-3/2}\big)  \Bigg).  \label{eq_derivative_t_density1}
\end{align}
Again using Lemma \ref{lemma_hn_relations} we get \eqref{eq_derivative_t_density}. 

Let us remark about the different powers in the error terms in the two cases. For the latter case, $\partial_{\bf t}\rho$, we need one less order in the expansion of $\ee^{-N\Omega}$ because the leading term, $-2\kappa Y={\cal O}(N^{\nu-1/2})$, is already of less order than the leading term of $\partial_{\bf n}\rho_n$.
\end{proof}

\begin{lemma}\label{lem:zzero} Let $z_0\in\partial K$. We have, for an arbitrary $\varepsilon>0$,
\begin{equation*} 
\rho_n(z_0)=\frac{1}{2\pi} - \frac{\kappa }{ 3 \sqrt{2\pi^3 N} }  + \mathcal{O}\big(N^{-3/2+\varepsilon}\big),
\end{equation*}
uniformly over $z_0\in\partial K$.
\end{lemma}

\begin{proof}
Let $0<\nu<1/6$. Choose $\tau$ such that $0<\tau<\nu$ and let
$$z_1 = z_0 - N^{-1/2+\tau} {\bf n}, $$
for sufficiently large $N$ such that $z_1\in {\rm Int}\,{K}$.
\begin{align} \nn 
\rho_n(z_0)-\rho_n(z_1)=\int_{-N^{-1/2+\tau}}^0 \partial_{\bf n} \rho_n(z_0 + X {\bf n}) \,\dd X.
\end{align}
Proposition \ref{prop_derivative_density} gives a uniform estimate of $\partial_{\bf n}\rho_n$ on the domain of integration.  We can replace $\partial_{\bf n}\rho_n$ by \eqref{eq_derivative_n_density}.
Then we only need to perform the integrals of the type 
\begin{equation*}
   \int_{-N^{-1/2+\tau}}^0 X^j \ee^{-2N X^2} \dd X =\frac{1}{2} \frac{(-1)^j}{(2N)^{(j+1)/2}} \Gamma \left(\frac{j+1}{2}\right) + {\cal O}\big(N^{-(j+1)/2+\tau(j-1)}\ee^{-2N^{2\tau}}\big),\quad j\in\NN\cup\{0\}.
 \end{equation*}
The error term is obtained by substituting $\widetilde X = \sqrt N X $ in the integral and using l'H\^opital's rule,
\begin{align*}
-\lim_{N\rightarrow \infty} \frac{\int_{-\infty}^{-N^{-1/2+\tau}}  X^j \ee^{-2N X^2} \dd X}{ N^{k-(j+1)/2} \ee^{-2N^{2\tau}}}&= -\lim_{N\rightarrow \infty} \frac{\int_{-\infty}^{-N^{\tau}}  \widetilde X^j \ee^{-2 \widetilde{X}^2} \dd \widetilde X}{ N^{k} \ee^{-2N^{2\tau}}}   = \lim_{N\rightarrow \infty} \frac{(-1)^j \tau N^{\tau(j+1)-1}  \ee^{-2 N^{2\tau}} }{ (k N^{k-1} -4 \tau N^{2\tau+k-1}  ) \ee^{-2N^{2\tau}}}\\
&= \frac{(-1)^{j+1}}{4} \lim_{N\rightarrow \infty} N^{\tau(j-1)-k} ,
\end{align*}
which converges to a non-zero constant if $k=\tau(j-1)$.

The straightforward calculation of the integration of \eqref{eq_derivative_n_density} gives 
\begin{align}
\rho_n(z_0)-\rho_n(z_1)&=- \frac{1}{2\pi} - \frac{\kappa }{ 3 \sqrt{2\pi^3 N} } 
+\mathcal{O}(N^{-3/2+9\nu}).\nn 
\end{align}
Note that the term of order $1/N$ vanishes. We will discuss a deeper reason for this in Section \ref{sec_consistency}.

The theorem is proved because, according to Theorem \ref{thm:1b}, $\rho_n(z_1)$ is $1/\pi$ up to an error given by ${\cal O}\big(N^{5/6} \ee^{-N^{2\tau} }\big) $.   The uniformity of the bound over $z_0\in\partial K$ follows from the uniformity of Proposition \ref{prop_derivative_density}. 
\end{proof}
We will prove the following slightly stronger statement than Theorem \ref{thm1}.

\vspace{0.2cm}
\begin{thm}\label{thm_density}
Let $z_0\in\partial K$ and 
\begin{equation*}
z-z_0 = {\bf n}\frac{\xi+i\eta }{\sqrt{N}},\qquad \xi,\eta\in\RR.	
\end{equation*}
For any $0 < \nu<1/6$ the asymptotic expansion
\begin{align}
\rho_n\left(z_0+\frac{(\xi+i\eta) {\bf n}}{\sqrt{N}}\right)-\frac{\erfc(\sqrt{2} \xi) }{2\pi} &=\frac{\ee^{-2 \xi^2} }{\sqrt{2\pi^3} }\Bigg\lgroup    \frac{\kappa}{ \sqrt{ N} } \left(\frac{\xi^2-1}{3}-\eta^2 \right)  + \frac{1}{N} \Psi(\xi,\eta) \Bigg\rgroup+ {\cal O}(N^{9\nu-3/2}) \nn 
\end{align}
holds uniformly over $\{z~|~|z-z_0|<N^{\nu-1/2}\}$ and over $z_0\in\partial K$, where we define
\begin{equation*}
\Psi(\xi,\eta)= \kappa^2  \frac{ 2\xi^5 - 8 \xi^3 + 3\xi +36 \xi\eta^2  -12 \xi^3 \eta^2  + 18 \xi\eta^4}{18}+ \left(\frac{ (\partial_s \kappa)^2 }{9 \kappa^2} - \frac{ \partial_s^2 \kappa }{12 \kappa} \right) \xi  + \frac{\partial_s \kappa}{3}  \Big( \xi^2 \eta  -\eta^3  - \eta \Big).
\end{equation*}
\end{thm}

\begin{proof}
The proof is by direct calculation where we integrate $\partial_{\bf n}\rho_n$ from $z_0$ to $z_0+{\bf n}\,\xi/\sqrt N$ and use Lemma \ref{lem:zzero} for the value of $\rho_n(z_0)$.  More explicitly we integrate $\partial_{\bf n}\rho_n$ using the expression in \eqref{eq_derivative_n_density}, over the integration variable $X$ from $0$ to $\xi/\sqrt N$, i.e.  
\begin{equation*}
	\rho_n \left(z_0+\frac{\xi {\bf n}}{\sqrt{N}}\right) - \rho_n(z_0)=\int_{0}^{\xi/\sqrt{N}} \partial_{\bf n}\rho_n(z_0+ X {\bf n}) \,\dd X.
\end{equation*}
Once we obtain the value of $\rho_n$ at $z_0 + \xi\,{\bf n}/\sqrt N$, we integrate $\partial_{\bf t}\rho_n$ from $z_0 + \xi\,{\bf n}/\sqrt N$ to the final destination $z_0 + \xi\,{\bf n}/\sqrt N + \eta\,{\bf t}/\sqrt N$.   Again, this is done by integrating the expression in \eqref{eq_derivative_t_density} over the integration variable $Y$ from $0$ to $\eta/\sqrt N$, i.e.
\begin{equation*}
	\rho_n \left(z_0+\frac{\xi {\bf n}}{\sqrt{N}}+\frac{\eta {\bf t}}{\sqrt{N}}\right) - \rho_n\left(z_0+\frac{\xi {\bf n}}{\sqrt{N}}\right)=\int_{0}^{\eta/\sqrt{N}} \partial_{\bf n}\rho_n\left(z_0+\frac{\xi {\bf n}}{\sqrt{N}}+Y{\bf t}\right) \,\dd Y.
\end{equation*}
After the straightforward calculation, the error bound comes from that of \eqref{eq_derivative_n_density}. 
\end{proof}

\section{Vanishing Cauchy Transform}\label{sec:vanishing-cauchy}

In this section, we consider a general potential $V$ that is defined in the sentence containing \eqref{eq:generalV}.

\vspace{0.3cm}
\begin{lemma}\label{lemma_density_kernel1}
For a fixed $z\in\CC$, let the kernel $K_n$ \eqref{eq:Kernel} satisfy the asymptotic behavior,
\begin{equation}\label{eq90}
|K_n(z,w)|=\frac{N}{\pi}e^{-\frac{N}{2}|z-w|^2}\Big(1+{\mathcal O}\big(e^{-\epsilon N}\big)\Big), 
\end{equation}
uniformly over $w\in B_\delta(z)$ for some $\epsilon>0$ and $\delta>0$. 
Then there exists $\epsilon'>0$ such that
\begin{equation*}
\int_{{\mathbb C}\setminus B_\delta(z)} |K_n(z,w)|^2 \dd A(w)={\mathcal O}\big(\ee^{-N \epsilon'}\big).
\end{equation*}
\end{lemma}
\begin{proof}
From the reproducing property of the kernel we have 
\begin{align}\label{eq:KnintK}\begin{split}
\int_{{\mathbb C}\setminus B_\delta(z)} |K_n(z,w)|^2 \dd A(w)&= \int_{\CC} |K_n(z,w)|^2\,\dd A(w) -\int_{ B_\delta(z)} |K_n(z,w)|^2\,\dd A(w)
\\ & =  K_n(z,z) -\int_{ B_\delta(z)} |K_n(z,w)|^2\,\dd A(w).
\end{split}
\end{align}
Taking $z=w$ in \eqref{eq90} we have
\begin{equation*}
|K_n(z,z)|=K_n(z,z)=\frac{N}{\pi}\Big(1+\mathcal{O}\big(\ee^{-N \epsilon}\big)\Big).
\end{equation*}
Integrating $|K_n(z,w)|$ over the disk $B_\delta(z)$ we have, for some $\epsilon_1>0$,
\begin{align*}
\int_{B_\delta(z)} |K_n(z,w)|^2 \dd A(w)&= \frac{N^2}{\pi^ 2}\int_{B_\delta(z)} \ee^{-N |z-w|^2} \left(1+\mathcal{O}\big(\ee^{-N\epsilon}\big)\right) \dd A(w)
\\&=\frac{N}{\pi}\Big(1+\mathcal{O}\big(\ee^{-N\delta^2}\big)\Big)\Big(1+\mathcal{O}\big(\ee^{-N\epsilon}\big)\Big).
\end{align*}
Proof is done by substituting the above two estimates into \eqref{eq:KnintK}.
\end{proof}

\begin{lemma}\label{lemma_density_kernel}
For a fixed $z\in\CC$, let Q be analytic in a neighborhood of $z$ and let the kernel $K_n$ \eqref{eq:Kernel} satisfy the asymptotic behavior,
\begin{equation}\label{eq9}
 \frac{1}{N} K_{n}(z,w)= \displaystyle\frac{1}{\pi}e^{-\frac{N}{2}|z-w|^2}e^{iN\, {\rm Im}\big(z\cc{w}-Q(z)- \overline{Q(w)} \big) }
\big(1+\mathcal{O}(\ee^{-N\epsilon})\big),
\end{equation}
uniformly for $w\in B_\delta(z)$ for some $\epsilon>0$ and $\delta>0$. 
Then there exists $\epsilon'>0$ such that
\begin{equation*}
\partial K_n(z,z)={\mathcal O}\big(\ee^{-\epsilon' N}\big). 
\end{equation*}

\end{lemma}
\begin{proof}
We define the pre-kernel for a general $Q$ by  
\begin{align*}
{\cal K}_n(z,w)=\sum_{j=0}^{n-1} p_j(z)\cc{p_j(w)} \ee^{-N(z\cc{w}-Q(z)-\cc{Q(w)})}
\end{align*}
which is analytic in $z$ and anti-analytic in $w$. Since ${\cal K}_n(z,z)=K_n(z,z)$ for all $z$, we have
 $$\partial K_n(z,z)=\partial {\cal K}_n(z,z)=\bigg[\frac{\partial}{\partial z} {\cal K}_n(z,w)\bigg]_{w=z}.$$
From the relation between ${\cal K}_n$ and $K_n$ (see \eqref{eq:pre}) and from \eqref{eq9}, we have that
\begin{equation*}
  {\cal K}_n(z,w) = \frac{N}{\pi}\Big(1+{\mathcal O}\big(e^{-\epsilon N}\big)\Big),\quad w\in B_\delta(z).
  \end{equation*}
Using the analyticity of ${\cal K}_n(z,w)$ in $z$, choosing ${\cal C}$ to be a circle around $z$ with the radius $\delta/2$, we have that
\begin{equation*}
\bigg[\frac{\partial}{\partial z} {\cal K}_n(z,w)\bigg]_{w=z}=\frac{1}{2\pi i} \oint_{\cal C}\frac{{\cal K}_n(\zeta,z)}{(\zeta-z)^2}d\zeta = {\cal O}\big( e^{-\epsilon N}\big).
\end{equation*}
Here we use that ${\cal K}_n(\zeta,z)=\overline{{\cal K}_n(z,\zeta)}$.
\end{proof}

\begin{proof}[Proof of Theorem \ref{thm:Cauchy}]
We recall that the probability density function for $n$ particles is given by 
\begin{equation*}
{\rm Prob}_n\big(\{z_1,\cdots,z_n\}\big)   = \frac{1}{{\cal Z}_n} \bigg(\prod_{1\le l<k\le n}|z_l-z_k|^2\bigg) \prod_{j=1}^n \ee^{-N V(z_j)}.
\end{equation*}
For $1\le m\le n$ the $m$-point correlation function is defined as 
$$R_n^{(m)}(z_1,\ldots z_m)= \frac{n!}{(n-m)!} \int _{\CC^{n-m}} {\rm Prob}_n\big(\{z_1,\cdots,z_n\}\big) \prod_{j=m+1}^n \dd A(z_j).$$
We note that $R_n^{(1)}(z) =K_n(z,z)=N\rho_n(z)$ and $R_n^{(2)}(z,w)= K_n(z,z) K_n(w,w)-|K_n(z,w)|^2$, see for example \cite{mehta2004random}.

Taking the holomorphic derivative of $\rho_n(z)$ one gets
\begin{align}\nonumber
\partial \rho_n(z)&= \frac{n}{N}\int_{\CC^{n-1}} \Big(-N\partial V(z) +\sum_{j=2}^n \frac{1}{z-z_j}\Big)\, {\rm Prob}_n\big(\{z,z_2,\cdots,z_n\}\big)\,\prod_{j=2}^n \dd A(z_j)
\\\nonumber &=-N \partial V(z) \rho_n(z) + \frac{n(n-1)}{N}\int_{\CC^{n-1}} \frac{\dd A(w)}{z-w} {\rm Prob}_n\big(\{z,w,z_3,\cdots,z_n\} \,\prod_{j=3}^n
dA(z_j) \\
\nonumber &=-N \partial V(z) \rho_n(z) + \frac{1}{N}\int_\CC \frac{\dd A(w)}{z-w}\Big( K_n(z,z)K_n(w,w)-|K_n(z,w)|^2\Big).
\end{align}
In the second equality we used that our probability density function is symmetric under permutation. Dividing by $\rho_n$ we get
\begin{equation}\label{eq1}
\frac{\partial \rho_n(z)}{\rho_n(z)}=-N \partial V(z) +\int_\CC\frac{K_n(w,w)}{z-w} \,\dd A(w) -\frac{1}{K_n(z,z)}\int_\CC\frac{|K_n(z,w)|^2}{z-w}
\,\dd A(w).\end{equation}
From $\rho_n(z)=\frac{1}{N}K_n(z,z)$, Lemma \ref{lemma_density_kernel} says that the left-hand side of the above equation is exponentially small in $N$. 

Below we show that the last term of the right-hand side of \eqref{eq1} is also exponentially small in $N$. 
\begin{align}\label{eq_int_split}
\int_\CC\frac{|K_n(z,w)|^2}{z-w}\,\dd A(w)=\int_{B_\delta(z)}\frac{|K_n(z,w)|^2}{z-w}\,\dd A(w)+\int_{\CC\setminus B_\delta(z)} \frac{|K_n(z,w)|^2}{z-w}\,\dd A(w).
\end{align}
For the first term we find 
\begin{align*}
\int_{B_\delta(z)}\frac{|K_n(z,w)|^2}{z-w}\,\dd A(w)&=\frac{N^2}{\pi^2}\int_{B_\delta(0)}\frac{1}{\zeta}\,\ee^{-N |\zeta|^2}\Big(1+{\mathcal O}\big(\ee^{-\epsilon N}\big)\Big) \dd A(\zeta) 
\\&=\frac{N^2}{\pi^2}\int_0^\delta\int_0^{2\pi} e^{-i\theta}\,\ee^{-N r^2}\Big(1+{\mathcal O}\big(\ee^{-\epsilon N}\big)\Big) \dd \theta\dd r =\mathcal{O}\big(N^{3/2}\ee^{-\epsilon N} \big).
 \end{align*}
For the second term in \eqref{eq_int_split} Lemma \ref{lemma_density_kernel1} says that, for some $\epsilon'>0$,
\begin{align*}
\left| \int_{\CC\setminus B_\delta(z)} \frac{|K_n(z,w)|^2}{z-w}\,\dd A(w) \right|& \le \frac{1}{\delta} \int_{\CC\setminus B_\delta(z)} |K_n(z,w)|^2\,\dd A(w)= \mathcal{O}\big(\ee^{-N \epsilon'} \big).
\end{align*}
We write the equation \eqref{eq1} replacing $n$ by $n+1$:
\begin{equation}\label{eq2}
\frac{\partial\rho_{n+1}(z)}{\rho_{n+1}(z)}=-N \partial V(z) +\int_\CC\frac{K_{n+1}(w,w)}{z-w} \, \dd A(w) -\frac{1}{K_{n+1}(z,z)}\int_\CC\frac{|K_{n+1}(z,w)|^2}{z-w}\,\dd A(w).\end{equation}
By the same arguments, the left hand side and the last term in the right hand side both vanish exponentially fast in $N$.
Subtracting \eqref{eq_int_split} from \eqref{eq2}, the term $-N \partial V(z)$ cancels, and one obtains that
\begin{equation*}
\int_\CC\frac{K_{n+1}(w,w)}{z-w}\,\dd A(w)-\int_\CC\frac{K_{n}(w,w)}{z-w}\,\dd A(w) =  \int_\CC\frac{|p_n(w)|^2\, \ee^{-NV(w)} }{z-w} \dd A (w)
\end{equation*}
vanishes exponentially fast as $N$ goes to infinity.
\end{proof}

\section{Discussions}\label{sec_discussion}
\subsection{
Implication of Theorem \ref{thm1}}\label{sec_consistency}

It is immediate that Theorem \ref{thm1} follows from Theorem \ref{thm_density}.  On the other hand, a geometric intuition suggests that the converse is also true, i.e., the density in the normal direction from the boundary determines the density in any direction.  In this section, we calculate the density in the tangential direction from the density in the normal direction and from the local shape of the curve $\partial K$.  This means that Theorem \ref{thm_density} and Theorem \ref{thm1} follow from each other.

Let $\Gamma$ be a smooth Jordan curve and $\gamma:\RR\to\Gamma$ ($s\mapsto\gamma(s)$) be the arclength parametrization.  Defining ${\bf n}$ to be the normal vector 
\begin{equation*}
  {\bf n}(\gamma(s)):= -i\frac{\dd\gamma}{\dd s},
\end{equation*}
let us consider the following mapping. 
\begin{equation*}
 \Xi:\RR^2\to \CC ,\quad \Xi(x,s)=  \gamma(s)+x\,{\bf n}(\gamma(s)).
\end{equation*}
The Jacobian of this mapping can be evaluated to be
\begin{equation}\label{eq:Jacobian}
1+\kappa(\gamma(s))\, x
\end{equation}
where $\kappa(\gamma(s))$ is the curvature of $\Gamma$ at $\gamma(s)$.  For a sufficiently small $|x|$ the Jacobian is strictly positive and, therefore, $\Xi$ is locally invertible around $\gamma(s)$ for any $s\in\RR$.  It means that there exists a neighborhood of $\gamma(s)$ such that $\Xi^{-1}$ defines a coordinate system on that neighborhood.

Let $z_0=\gamma(s_0)\in\Gamma$.  We have $\Xi(0,s_0)=z_0$. Since $\Xi$ is locally a diffeomorphism, there exists a neighborhood $U_0$ of $z_0$ and a neighborhood $V_0\subset\RR^2$ of the origin such that
\begin{equation*}
\Xi_0:V_0\to U_0,\quad (x,y)\mapsto \Xi_0(x,y):=\Xi(x,s_0+y)
\end{equation*}
is a diffeomorphism.

We define the mapping ${\cal F}_0:\RR^2\to \CC$ by
\begin{equation*}
(X,Y)\mapsto {\cal F}_0(X,Y):=z_0 + X{\bf n}(z_0) + Y {\bf t}(z_0),\quad {\bf t}(z_0) = i{\bf n}(z_0).
\end{equation*}
It is clear that ${\cal F}_0$ defines a {\em flat} coordinate around $z_0$.

The relation between the two coordinates, $(X,Y)$ and $(x,y)$, is given by the following proposition.

\vspace{0.3cm}
\begin{prop}\label{prop_geom}
For $(x,y)\in V_0$, there is a unique $(X,Y)\in \RR^2$ such that $$\Xi_0(x,y)={\cal F}_0(X,Y).$$
As $z$ goes to $z_0$, the two coordinates $(x,y)$ and $(X,Y)$ are related by
\begin{align*}
x&= X + \frac{\kappa }{2}Y^2 + \frac{\partial_s \kappa}{6} Y^3 - \frac{ \kappa^2 }{2}X Y^2+ \frac{\kappa^3}{2}X^2Y^2 -\frac{\kappa\partial_s \kappa}{2}XY^3 + \left(\frac{\partial_s^2\kappa}{24} -\frac{\kappa^3}{8}\right)Y^4   +\mathcal{O}(|z-z_0|^5),
\\\nn
y&=Y -\kappa XY + \kappa^2 X^2Y -\frac{\partial_s\kappa}{2} XY^2 -\frac{\kappa^2}{3} Y^3 
\\ & \qquad\qquad\qquad-\kappa^3 X^3 Y + \frac{3\kappa\partial_s\kappa}{2}X^2 Y^2-\frac{7\kappa\partial_s\kappa}{24} Y^4 +\left(\kappa^3-\frac{\partial_s^2\kappa}{6}\right) X Y^3+ \mathcal{O}(|z-z_0|^5), 
\end{align*}
where $\kappa$ and its derivatives are all evaluated at $z_0\in\Gamma$.
\end{prop}

\begin{proof}
A Taylor expansion gives
\begin{align*}
  \gamma(y+s_0) &= z_0+\gamma'(s_0) y+\frac{1}{2}\gamma''(s_0) y^2+\frac{1}{6}\gamma'''(s_0) y^3+\frac{1}{24}\gamma^{(4)}(s_0) y^4+{\cal O}(|y|^5)
  \\
&  =z_0+{\bf t} y-\frac{\kappa}{2}{\bf n} y^2-\frac{1}{6}(\partial_s\kappa\, {\bf n} +\kappa^2{\bf t}) y^3
  -\frac{1}{24}(\partial_s^2\kappa\, {\bf n}+3\kappa\partial_s\kappa{\bf t}-\kappa^3{\bf n}) y^4+{\cal O}(|y|^5),
\end{align*}
where, in the last expression, the Taylor coefficients are evaluated at $z_0$.
In the last equality we used the following elementary facts.
\begin{equation*}
  \partial_s\gamma={\bf t},\quad \partial_s{\bf n}=\kappa{\bf t},\quad\partial_s{\bf t}=-\kappa{\bf n},
\end{equation*}
where ${\bf n}$ and ${\bf t}$ are understood to depend on $\gamma(s)$.
Similarly, one can get
\begin{align*}
  {\bf n}(\gamma(y+s_0))&={\bf n}(z_0)+\partial_s{\bf n}(z_0) y +\frac{1}{2}\partial_s^2{\bf n}(z_0) y^2+\frac{1}{6}\partial_s^3{\bf n}(z_0) y^3 +{\cal O}(|y^4|)
  \\
  &=
  {\bf n}+\kappa{\bf t} y +\frac{1}{2}(\partial_s\kappa {\bf t} - \kappa^2 {\bf n}) y^2+\frac{1}{6}(\partial_s^2\kappa {\bf t}-3 \kappa\partial_s\kappa  {\bf n} -\kappa^3 {\bf t} ) y^3 +{\cal O}(|y^4|),
\end{align*}
where all the coefficients are evaluated at $z_0$.

Then we can expand $\Xi_0(x,y)$ as follows.
\begin{equation}\label{Xi-expand}
\begin{split}
\Xi_0(x,y)&= z_0 + {\bf n}(z_0)\bigg(x-\frac{\kappa}{2}y^2-\frac{\kappa^2}{2}x y^2-\frac{\partial_s\kappa}{6}y^3 +\frac{1}{24}(\kappa^3-\partial_s^2\kappa)y^4 -\frac{1}{2}(\kappa\partial_s\kappa)xy^3 \bigg) 
\\
&\qquad + {\bf t}(z_0)\bigg( y +\kappa xy+\frac{1}{2}\partial_s\kappa x y^2-\frac{\kappa^2}{6} y^3-\frac{\kappa\partial_s\kappa}{8}y^4+\frac{1}{6}(\partial_s^2\kappa -\kappa^3)xy^3 
\bigg)+ {\cal O}(|z-z_0|^5)
\\& = z_0 + {\bf n}(z_0) X + {\bf t}(z_0) Y. 
\end{split}
\end{equation}
So we have obtained the series expansions of $X$ and $Y$ in terms of $x$ and $y$.   To find the series expansion of $x$ and $y$ in terms of $X$ and $Y$, we simply substitute the most general Taylor expansions, 
\begin{equation*}
  x= \sum_{j+k\leq 4} c_{jk} X^j Y^k  + {\cal O}(|z-z_0|^5),\qquad 
    y= \sum_{j+k\leq 4} d_{jk} X^j Y^k + {\cal O}(|z-z_0|^5),
\end{equation*}
into \eqref{Xi-expand} and determine the coefficients $c_{jk}$'s and $d_{jk}$'s by matching each order.
\end{proof}

Below we illustrate how Theorem \ref{thm_density} is obtained from Theorem \ref{thm1}.  We only consider $\xi=0$ (i.e. purely tangential direction) to make the presentation simple.

For $z_0=\gamma(s_0)$, we define $x$ and $y$ to satisfy the following relation. 
\begin{equation*}
  z_0+\frac{\eta\,{\bf t}}{\sqrt N} = \Xi_0(x,y).
\end{equation*}
Using the above $x$ and $y$, we define $\widehat z_0:=\gamma(s_0+y)$.  Then the identity,
\begin{equation}\label{eq:tangent-in-normal}
  \rho_n \left(z_0+\frac{\eta {\bf t}}{\sqrt{N}}\right) = \rho_n \left(\widehat z_0+ x\,{\bf n}(\widehat z_0)\right),
\end{equation}
is simply the density at the same point expressed in two different coordinate systems. 
This equation gives the density in the tangential direction expressed in terms of the density in the normal direction.  This already implies that Theorem \ref{thm1} gives Theorem \ref{thm_density} as we mentioned in the beginning of this section.

Proposition \ref{prop_geom} gives the asymptotic expressions of $x$ and $y$ in terms of $\eta$.
\begin{equation*}
  x= \frac{\kappa}{2}\frac{\eta^2}{N}+\frac{\partial_s\kappa}{6}\frac{\eta^3}{N^{3/2}}+{\cal O}\left(\frac{|\eta|^4}{N^2}\right), \qquad y= \frac{\eta}{\sqrt N} -\frac{\kappa^2}{3} \frac{\eta^3}{N^{3/2}}+{\cal O}\left(\frac{|\eta|^4}{N^2}\right).
\end{equation*}
Using these series expansions with
\begin{equation*}
  \kappa(\widehat{z_0})=\kappa(z_0)+\partial_s\kappa(z_0)\frac{\eta}{\sqrt N} + {\cal O}\left(\frac{1}{N}\right)
\end{equation*} 
and the scaled normal coordinate:
$$\widehat\xi:=\sqrt N x=\frac{\kappa\,\eta^2}{2\sqrt N} + \frac{\partial_s\kappa}{6}\frac{\eta^3}{N} +{\cal O}\left(\frac{1}{N^{3/2}}\right),$$
Theorem \ref{thm1}
gives 
 \begin{align}\label{eq:tan-density}
 \begin{split}
\rho_n \left(z_0+\frac{\eta {\bf t}}{\sqrt{N}}\right) &=  \rho_n\left(\widehat{z_0}+\frac{\widehat\xi}{\sqrt N} {\bf n}(\widehat z_0)\right)
  =\frac{1 }{2\pi}\erfc\left(\sqrt 2\widehat\xi \right) +\frac{\ee^{-2 \hat{\xi}^2} }{  \sqrt{2\pi^3} }  \frac{\kappa(\widehat{z_0})}{ \sqrt{ N} } \left(\frac{\widehat\xi^2-1}{3} \right)  + \mathcal{O}\left(\frac{\widehat \xi}{N}\right) 
  \\
  &=\frac{1}{2\pi } -\frac{1}{\sqrt{2\pi^3}} \bigg(\frac{ \kappa(z_0)(1+ 3\eta^2)}{3\sqrt{N}} 
+\frac{ \partial_s \kappa(z_0) ( \eta^3  + \eta) }{3 N }\bigg) + o\left(\frac{1}{N}\right).
\end{split}
 \end{align}
This is exactly what Theorem \ref{thm_density} says for $\xi=0$.  

\subsection{Orthogonal polynomials from density}

We have obtained the asymptotic expansion of the density using the asymptotics of the orthogonal polynomials.  We can also do the opposite, i.e., we can obtain the subleading correction of the orthogonal polynomials assuming certain behavior of the density.  
In fact, we obtained the expression in Lemma \ref{lemma_hn_relations} by this idea. The proof in the appendix is not very motivating as it does not indicate how we find the identities in the first place.  
Below we show another proof of Lemma \ref{lemma_hn_relations} that does not require any prior knowledge about $h_{-1}$ and $h_0$.  Instead, we will use the behavior of the density given in Theorem \ref{thm1} and Theorem \ref{thm:1b} (or its weaker version Lemma \ref{lem:density}).

\begin{proof}[Alternative proof of Lemma \ref{lemma_hn_relations}]
One can integrate \eqref{eq_derivative_t_density1} in the tangential direction to get
\begin{align*}
\rho_n \left(z_0+\frac{\eta {\bf t}}{\sqrt{N}}\right)&= \rho_n\left(z_0\right)- \frac{1}{\sqrt{2\pi^3}}\bigg(\frac{\kappa \eta^2}{\sqrt{N}} -\frac{1}{ N}\bigg(2\kappa \eta\frac{ {\rm Im}\left(h_{-1}(z_0)-h_0(z_0)\right) }{|\psi'(z_0)|}- \frac{\eta^3 \partial_s\kappa}{3}    \bigg)\bigg) + o\left(\frac{1}{N}\right)
\\
&=\frac{1}{2\pi}- \frac{1}{\sqrt{2\pi^3}}\bigg(\frac{\kappa (3\eta^2+1)}{3\sqrt{N}} -\frac{1}{ N}\bigg(2\kappa \eta\frac{ {\rm Im}\left(h_{-1}(z_0)-h_0(z_0)\right) }{|\psi'(z_0)|}- \frac{\eta^3 \partial_s\kappa}{3}    \bigg)\bigg) + o\left(\frac{1}{N}\right).
\end{align*}
Comparing the term of order $1/N$ with that of \eqref{eq:tan-density} we obtain the second equation in Lemma \ref{lemma_hn_relations}.

Let $z_0\in\partial K$ and $z_1\in{\rm Int}\,K$ be given as in the proof of Lemma \ref{lem:zzero}.  Let $0<\tau<\nu<1/6$.  Let $z_2= z_0 + N^{-1/2+\tau} {\bf n}$ such that it is located symmetrically to $z_1$ with respect to $\partial K$.   Using the same method as in the proof of Lemma \ref{lem:zzero}, one can integrate \eqref{eq:with-hn} to get  
\begin{equation*}
\begin{split}
\rho_n(z_1)-\rho_n(z_2)&=\frac{1}{\pi} 
+\frac{1}{\pi N}\bigg[{\rm Re}\Big(h_0(z_0)+h_{-1}(z_0)\Big)+\frac{\partial_s\kappa\, {\rm Im}\left(h_{-1}(z_0)-h_{0}(z_0)\right)}{3|\psi'(z_0)|\kappa}
\\& \qquad\qquad\qquad\qquad +\frac{\kappa^2}{12}  +\frac{(\partial_s\kappa)^2}{18 \kappa^2}-\frac{\partial_s^2\kappa}{24\kappa} \bigg] +\mathcal{O}(N^{-3/2+9\nu}).
\end{split}
\end{equation*}
From Theorem \ref{thm:1b} the right-hand side must be equal to $1/\pi$ up to an error given by ${\cal O}\big(N^{5/6}\ee^{-N^{2\tau}}\big)$.  Taking $\nu<1/18$, this implies that the $1/N$-term above vanishes.  This provides the first equation of Lemma \ref{lemma_hn_relations} when combined with the second equation that we already obtained.
\end{proof}

We remark that the two relations in Lemma \ref{lemma_hn_relations}, combined with the asymptotic behavior of $h_0(z)$ and $h_{-1}(z)$ as $z$ goes to $\infty$, are enough to fully determine $h_0$ and $h_{-1}$.
Lemma \ref{lemma_hn_relations} gives ${\rm Re}(h_0(z)+h_{-1}(z))$ and ${\rm Im}(h_0(z)-h_{-1}(z))$ on the ellipse. 
Since both functions are harmonic on the exterior of the ellipse, and since both are bounded at $\infty$, Lemma \ref{lemma_hn_relations} determines both functions on the exterior of the ellipse. This uniquely determines $h_0(z)+h_{-1}(z)$ up to an imaginary constant and $h_0(z)-h_{-1}(z)$ up to a real constant, on the exterior of the ellipse.  The undetermined constants can be fixed by the behavior of the orthonormal polynomials' leading coefficients $\gamma_{n}$ and $\gamma_{n-1}$ (see Lemma \ref{lem:rho}) when $n,N\rightarrow \infty$ while $n=NT$.

\subsection{The number of particles outside $K$}

In the scaling limit, most of the Coulomb particles are supported on $K$. A typical sample looks like in Figure \ref{fig_sample}. In this section we will evaluate the number of particles outside $K$.  This problem was suggested to us by Arno Kuijlaars.

\begin{figure}[t!]
\begin{center}
\includegraphics[width=0.9\textwidth]{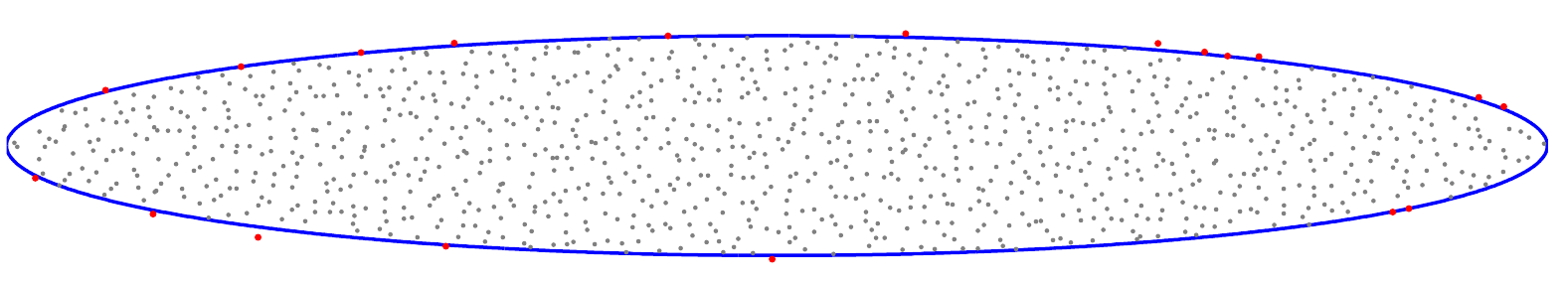}
\caption{Eigenvalue distribution of a sample of a random matrix $M_{n,t,T}$ \eqref{eq:MntT} with $n=1024$, $t=0.75$ and $T=1$}\label{fig_sample}
\end{center}
\end{figure}

Note that the total number of particles is given by $n=N \times {\rm Area}(K)/\pi$.
Therefore, the expected number of particles outside $K$ is given by 
\begin{equation*}
  n_\text{out} :=n - \#(\text{particles inside $K$}) = N \int_K \bigg(\frac{1}{\pi}-\rho_n(z)\bigg)\dd A(z).
\end{equation*}
Inside $K$ and away from $N^{\tau-1/2}$-neighborhood (that we denote by $U_{N,\tau}$ in Theorem \ref{thm:1b}) of $\partial K$, Theorem \ref{thm:1b} says that  $\pi^{-1}-\rho_n(z)$ disappears faster than any algebraic power of $N$ in the scaling limit.   This leads to
\begin{equation*}\begin{split}
  n_\text{out} &= N \int_{K\cap U_{\tau,N}} \bigg(\frac{1}{\pi}-\rho_n(z)\bigg)\dd A(z) +  \mathcal{O}( \ee^{-N^{\tau}}) 
 \\
& = \sqrt N \int_0^L \int_{-N^{\nu}}^0 \bigg[\frac{1}{\pi}-\rho_n\left( \gamma(s) +\frac{\xi}{\sqrt N}{\bf n}(\gamma(s))\right)\bigg]\left(1+\kappa(\gamma(s))\frac{\xi}{\sqrt N} \right) \dd\xi\,\dd s +  \mathcal{O}(\ee^{-N^{\tau}}) 
  \end{split}
\end{equation*}
where $L$ is the length of $\partial K$, $\gamma:[0,L)\to \partial K$ is the arclength parametrization of $\partial K$, and the factor $(1+\kappa\xi/\sqrt N)/\sqrt{N}$ is the Jacobian coming from the change of the measure from $\dd A$ to $\dd\xi\,\dd s$.

At a boundary point $z_0\in\partial K$, Theorem \ref{thm1} says that
\begin{align*}
\rho_{n}\bigg(z_0+\frac{\xi}{\sqrt N}{\bf n}\bigg)&=\rho^{(0)}_\infty(\xi)+\frac{1}{\sqrt N} \rho^{(1)}_\infty(\xi) + \frac{1}{N} \rho^{(2)}_\infty(\xi)+  {\cal O}\left(N^{-3/2+9\nu} \right),
\end{align*}
where
\begin{equation*}\label{eq:rho01}\begin{split}
\rho^{(0)}_\infty(\xi) &=\frac{1}{2\pi}\erfc(\sqrt{2} \xi), \\
\rho^{(1)}_\infty(\xi) &= \frac{\kappa(z_0)  }{  3\sqrt{2\pi^3} } (\xi^2-1)\,\ee^{-2 \xi^2},\\
\rho^{(2)}_\infty(\xi)&=\frac{\ee^{-2 \xi^2}}{\sqrt{2\pi^3}}\left(\kappa^2 \frac{ 2\xi^5 - 8 \xi^3 + 3\xi}{18}+\left(\frac{ (\partial_s \kappa)^2 }{9 \kappa^2} - \frac{ \partial_s^2 \kappa }{12 \kappa} \right) \xi\right).
\end{split}\end{equation*}
We get 
\begin{equation}\begin{split}\label{eq_evoutside}
  n_\text{out}&= \sqrt N\int_0^L \int_{-N^{\nu}}^0 \bigg[\frac{1}{\pi}-\rho^{(0)}_\infty(\xi)- \frac{1}{\sqrt N}\rho^{(1)}_\infty(\xi)-\frac{1}{N}\rho^{(2)}_\infty(\xi)\bigg]\left(1+\kappa(\gamma(s))\frac{\xi}{\sqrt N} \right) \dd\xi\,\dd s + o\left(\frac{1}{\sqrt N}\right)
  \\&=
  \sqrt N\int_0^L \int_{0}^{N^{\nu}} \bigg[\rho^{(0)}_\infty(\widetilde{\xi}) 
  -\frac{\kappa(\gamma(s))}{\sqrt N}\left( \widetilde{\xi} \rho^{(0)}_\infty(\widetilde{\xi}) + \frac{\widetilde{\xi}^2-1}{3\sqrt{2\pi^3}} \ee^{-2 \widetilde{\xi}^2}\right)\\
&\qquad\qquad\qquad\qquad\qquad\qquad\qquad\qquad\qquad  +\frac{1}{N}\left(\kappa(\gamma(s))\widetilde{\xi} \rho^{(1)}_\infty(\widetilde{\xi}) +  \rho^{(2)}_\infty(\widetilde{\xi}) \right)
  \bigg] \dd\widetilde{\xi}\,\dd s + o\left(\frac{1}{\sqrt N}\right)
\\&=
   \frac{\sqrt{N}}{(2\pi)^{3/2}} \int_0^L  \bigg[1 
  +0
  -\frac{1}{N}\left(\frac{\kappa^2(\gamma(s))}{12} - \frac{(\partial_s \kappa(\gamma(s)))^2}{18 \kappa^2(\gamma(s))}+\frac{\partial_s^2\kappa(\gamma(s))}{24 \kappa(\gamma(s))} \right) \bigg] \dd s + o\left(\frac{1}{\sqrt N}\right),\\
\end{split}
\end{equation}
where we have chosen $\nu$ such that $0<10\nu<1/2$. In the second equality, we made the substitution $\widetilde{\xi}=-\xi$ and used the identity and symmetries:
\begin{equation*}
  \rho_\infty^{(0)}(\xi) + \rho_\infty^{(0)}(-\xi)= \frac{1}{\pi}, \qquad \rho_\infty^{(1)}(-\xi)=\rho_\infty^{(1)}(\xi), \qquad \rho_\infty^{(2)}(-\xi)=-\rho_\infty^{(2)}(\xi).
\end{equation*}
One notices that the subleading term of order $N^{-1/2}$ in the penultimate formula for $n_\text{out}$ cancels out.  And the ``arclength density of escaping particles'' over $\partial K$ is uniform in the leading order. 

Parameterizing the ellipse by $\phi(\ee^{i\theta})$ with $\theta\in[0,2\pi)$ where the arclength element is given by $\dd s=|\phi'(\ee^{i\theta})|\dd \theta$ and using the explicit expressions for $\kappa$, $\partial_s \kappa$ and $\partial^2_s \kappa$ (see \eqref{eq:kappa0}, \eqref{eq:kappa1} and \eqref{eq:kappa2} respectively), the integral over the arclength in \eqref{eq_evoutside} can be evaluated to give
\begin{align}\nn
n_\text{out}&=\frac{\sqrt{n}}{(2\pi)^{3/2}} \int_0^{2\pi}\left(1+\frac{t^6+ 9 t^4- 9 t^2-1-6t(t^4-1) \cos(2 \theta)}{12n (1 + t^2 - 2 t \cos(2 \theta))^3} \right)\frac{\sqrt{1+ t^2 - 2 t \cos(2 \theta)}}{\sqrt{1-t^2}} \dd\theta +o\left(\frac{1}{\sqrt n}\right) \\
 &=  \frac{\sqrt n }{(2\pi)^{3/2}}\left( \frac{4(1+t)^{1/2} {\rm E}(k)}{(1-t)^{1/2} } -\frac{1}{n}\frac{ (1+t^2) {\rm E}(k) +  2(1-t)^2 {\rm K}(k) }{9 (1-t)^{3/2} (1+t)^{1/2} }  \right) + o\left(\frac{1}{\sqrt n}\right),  
  \label{eq_evoutside2}
\end{align}
with $k=\frac{2\sqrt{t}}{1+t}$. $\rm K(k)$ and $\rm E(k)$ are the complete elliptic integrals of the first and second kind,
\begin{equation*}
{\rm K}(k)= \int_{0}^{\pi/2}\frac{\dd \theta}{\sqrt{1-k^2 \sin^2(\theta)}}, \qquad {\rm E}(k)=\int_{0}^{\pi/2}\sqrt{1-k^2 \sin^2(\theta)}\dd \theta.
\end{equation*}
Note that we can write the leading term also using $L$, the circumference of the ellipse, 
\begin{equation*}
n_\text{out}=  \frac{\sqrt{n} L}{(2\pi)^{3/2} \sqrt{T}}+\mathcal{O}(n^{-1/2}).
\end{equation*}  
This expression is universally true, for any potential that shares the density profile $\rho_\infty^{(0)}$ at the boundary, when we define $L$ to be the length of $\partial K$.

The similar calculation can show that those ``escaping particles'' did not go far.
Below we calculate the expected number of particles outside $K$ in the $N^{\tau-1/2}$-neighborhood of $\partial K$:
\begin{align*}
N &\int_{U_{N,\tau}\setminus K} \rho_n(z)\dd A(z) = \sqrt N\int_0^L \int_{0}^{N^{\nu}} \bigg[\rho^{(0)}_\infty({\xi}) 
  +\frac{\kappa(\gamma(s))}{\sqrt N}\left( {\xi} \rho^{(0)}_\infty({\xi}) + \frac{{\xi}^2-1}{3\sqrt{2\pi^3}} \ee^{-2 {\xi}^2}\right)\\
&\qquad\qquad\qquad\qquad\qquad\qquad\qquad\qquad\qquad\qquad  +\frac{1}{N}\left(\kappa(\gamma(s)){\xi} \rho^{(1)}_\infty({\xi}) +  \rho^{(2)}_\infty({\xi}) \right)
  \bigg] \dd{\xi}\,\dd s + o\left(\frac{1}{\sqrt N}\right).
\end{align*}
Comparing with the second line of \eqref{eq_evoutside}, we notice that the ``arclength density of the leaked particles'' over $\partial K$ is exactly same as that of the escaping particles since the two integrands only differs by a minus sign in front of the term of order $1/\sqrt N$ that turned out to be zero.

Based on the observations in this section one may say that, when $N$ is finite, (the normalized measure of) the particles are shifted from the limiting distribution $\frac{1}{\pi}{\bf 1}_K$ at $N=\infty$, only around the boundary $\partial K$ and only in the {\em normal} directions of $\partial K$.

We can evaluate the number of particles outside the ellipse numerically using Matlab. We use the fact \cite{rider2003limit} that the particles distributed according to \eqref{Eq:gibbs} share the same distribution as the eigenvalues of the random matrix ensemble, given by 
\begin{equation}\label{eq:MntT}
  M_{n,t,T}=\frac{\sqrt{T}}{\sqrt{1-t^2}}\left(\sqrt{1+t} A_n+ i \sqrt{1-t} B_n\right)
\end{equation}
where $A_n$ and $B_n$ are independent copies of an $n\times n$ GUE, i.e. Hermitian matrices where the diagonal entries are i.i.d. ${\cal N}(0,1/(2n))$ and the upper triangular entries are i.i.d. ${\cal N}(0,1/(4n))+i{\cal N}(0,1/(4n))$. Our numerical calculation seems to support the asymptotic formula \eqref{eq_evoutside2}. We can see from Table \ref{table1} that the numbers in the last column tend to converge to zero as $n$ grows (except for a few higher $n$'s where the error bounds are too large to tell).

\begin{table}[bt]
  \centering
\begin{tabular}{c| r r| r r r} 
 \multicolumn{1}{c}{ } & \multicolumn{1}{c}{$n$} & \multicolumn{1}{c}{samples} &\multicolumn{1}{c}{$\overline{n}_\text{out}$}  & \multicolumn{1}{l}{${\mathbb E}_1$} & \multicolumn{1}{c}{$\sqrt{n}\left( {\mathbb E}_1 -\overline{n}_\text{out} \right)$ } \\   \hline
\multirow{6}{1.9cm}{$t=0.5$ \\ \vspace{\fill} $T=1$} & 2 & $2.41\cdot10^{8}$ & $0.66073\pm8\cdot10^{-5}$ & 0.66226 & $(21.7\pm1.2)\cdot10^{-4}$ \\ 
& 4 & $1.13\cdot10^{9}$ & $0.95741\pm5\cdot10^{-5}$  & 0.95822 & $(16.3\pm0.9)\cdot10^{-4}$ \\ 
& 8 & $7.80\cdot10^{8}$ & $1.36998\pm7\cdot10^{-5}$  & 1.37043 & $(12.8\pm1.9)\cdot10^{-4}$ \\ 
& 16 & $2.7\phantom{0}\cdot10^{9}$ & $1.94876\pm4\cdot10^{-5}$ & 1.94890 & $(6.0\pm1.8)\cdot10^{-4}$ \\ 
& 32 & $9.54\cdot10^{8}$ & $2.76381\pm9\cdot10^{-5}$  & 2.76382 & $(0.3\pm5.0)\cdot10^{-4}$ \\
& 64 & $1.65\cdot10^{8}$ & $3.9141\phantom{0}\pm3\cdot10^{-4}$  & 3.91404 & $(-4.1\pm20.2)\cdot10^{-4}$  \\ \hline
\multirow{4}{1.9cm}{$t=0.75$ \\ \vspace{\fill} $T=0.4375$}  & 4 & $2.16\cdot10^{8}$ & $1.3261\phantom{5}\pm1\cdot10^{-4}$ & 1.33977 & $(27.2\pm0.2)\cdot10^{-3}$ \\ 
& 10 & $1.33\cdot10^{8}$ & $2.1535\phantom{5}\pm2\cdot10^{-4}$ & 2.15936 & $(18.4\pm0.6)\cdot10^{-3}$ \\ 
& 32 & $3.69\cdot10^{8}$ & $3.8944\phantom{2}\pm2\cdot10^{-4}$  & 3.89639 & $(11.1\pm0.9)\cdot10^{-3}$ \\ 
& 64 & $2.43\cdot10^{8}$ & $5.5201\phantom{2}\pm3\cdot10^{-4}$  & 5.52114 & $(8.1\pm2.0)\cdot10^{-3}$ \\  
\end{tabular}
\caption{Numerical simulation of random matrices $M_{n,t,T}$. $\overline{n}_\text{out}$ is the number of eigenvalues outside the ellipse averaged over all samples, ${\mathbb E}_1$ the expectation value according to \eqref{eq_evoutside2} including the first subleading correction. In the $\overline{n}_\text{out}$ column and in the last column we show 95\% confidence intervals. }
  \label{table1}
\end{table}

\begin{figure}
\includegraphics[width=\textwidth]{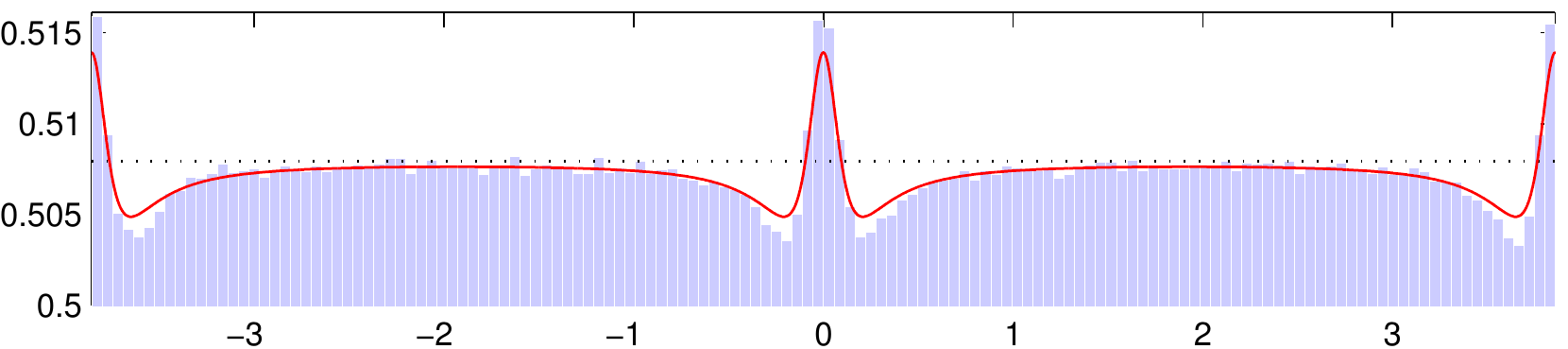}
\caption{Histogram of the distribution of outside eigenvalues regarding arclength where $t=0.5$, $T=1$, $n=64$, $\#$ of samples $=1.25 \cdot 10^8$. The histogram shows the arclength density of the outside eigenvalues. We divide the ellipse by 140 arcs of the same length, and counted the average numbers of ``escaping'' eigenvalues that are projected to each arc.  
The dotted line gives the leading theoretical behavior, while the solid line is the theoretical curve that includes the first subleading correction.  The $0$ corresponds to the point of the maximal curvature on the ellipse.
 \label{fig:evoverarclength}}
\end{figure}

We also looked at the distribution of the eigenvalues outside over the arclength of the ellipse. For this we projected the eigenvalues outside of the ellipse onto the ellipse, in a direction normal to the boundary, i.e. we associated an eigenvalue outside of the ellipse with its nearest point on the ellipse. See Figure \ref{fig:evoverarclength} for the plot of the numerical simulation. 

\appendix

\section{Plots of density: higher corrections}\label{app:plots}
Here we collect a few plots of $\rho_n$ across the boundary that supports Theorem \ref{thm1}. 

The leading behavior is shown in Figure \ref{fig-rho2}.  
In Figure \ref{fig:twoplots} we plot the subleading corrections of $\rho_n$, i.e.
\begin{align*}
\begin{split}
\rho^{(1)}_n(\xi)&:= \sqrt{N}\bigg[\rho_{n}\bigg(z_0+\frac{\xi}{N^{1/2}}{\bf n}\bigg) - \frac{1}{2\pi}\erfc(\sqrt{2} \xi) \bigg],
\\
\rho_n^{(2)}(\xi)&:=N\bigg[\rho_{n}\bigg(z_0+\frac{\xi}{N^{1/2}}{\bf n}\bigg) - \frac{1}{2\pi}\erfc(\sqrt{2} \xi) - \frac{\kappa}{\sqrt N}\frac{1  }{  3\sqrt{2\pi^3} } (\xi^2-1)\,\ee^{-2 \xi^2}\bigg].
\end{split}
\end{align*}
According to Theorem \ref{thm1} we know that
\begin{equation}\label{eq:dashed}
\begin{split}
\lim_{n\to\infty}	\rho^{(1)}_n(\xi) &= \frac{\kappa  }{  3\sqrt{2\pi^3} } (\xi^2-1)\,\ee^{-2 \xi^2},
\\
\lim_{n\to\infty} \rho^{(2)}_n(\xi)&=\frac{\ee^{-2 \xi^2}}{\sqrt{2\pi^3}}\left(\kappa^2 \frac{ 2\xi^5 - 8 \xi^3 + 3\xi}{18}+\left(\frac{ (\partial_s \kappa)^2 }{9 \kappa^2} - \frac{ \partial_s^2 \kappa }{12 \kappa} \right) \xi\right).
\end{split}
\end{equation}
\begin{figure}
\includegraphics[width=8cm]{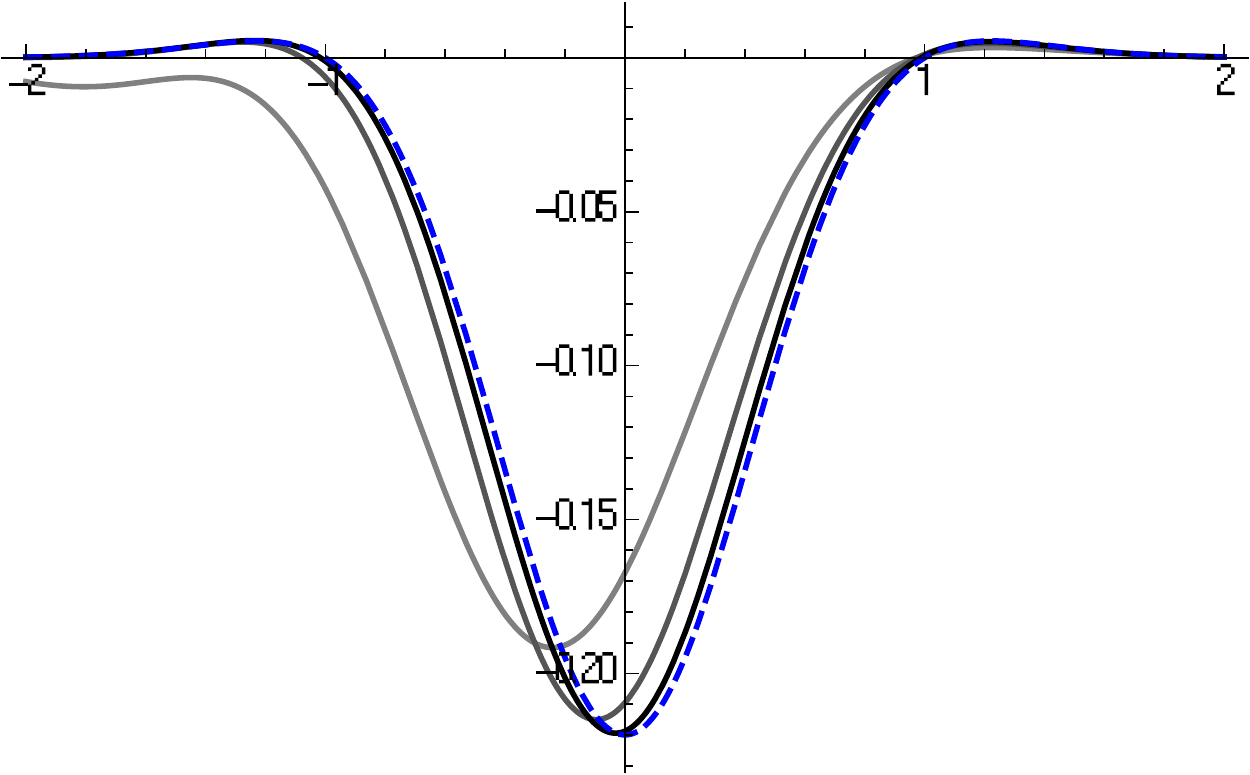}
\qquad
\includegraphics[width=8cm]{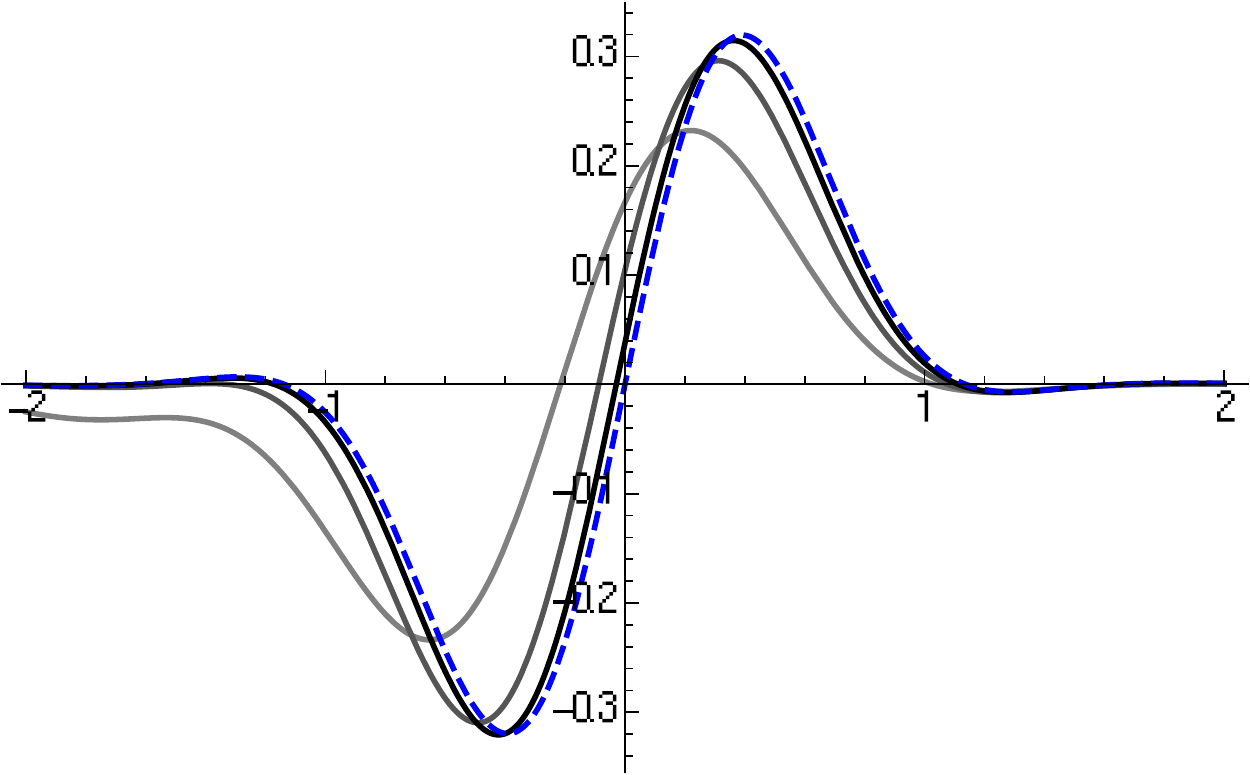}
\caption{$\rho^{(1)}_n$ (left) and $\rho^{(2)}_n$ (right) for $n=10, 100, 1000$ and $\infty$ (dashed lines) \label{fig:twoplots}}
\end{figure}
Both pictures are obtained for $t=0.5$, $T=1$ (i.e. $N=n$) and $z_0=\phi(1)\approx 1.73$.   We also mention that $\kappa(z_0)\approx 5.19$, $\partial_s\kappa(z_0)=0$ and $\partial_s^2\kappa(z_0)\approx -374$. 
The gray, darker gray and black lines correspond to $n=10$, $n=100$ and $n=1000$ respectively.  The dashed lines are the theoretical limits shown in \eqref{eq:dashed}.

\section{WKB approximation}\label{app:WKB}

The Hermite polynomials fulfill the differential equation
\begin{equation*}
H_k''(z)-2z H_k'(z)+2k H_k(z)=0,\quad k=0,1,2,\cdots.
\end{equation*}
By Theorem \ref{prop:Hermite}, the orthonormal polynomial, $p_k$, satisfies 
\begin{equation}\label{eq_ode}
\frac{F_0^2}{2NT} p_k''(z)-2z\, p_k'(z) +2k\, p_k(z)=0,\quad k=0,1,2,\cdots.
\end{equation}
We note that $n=NT$.

In \cite{dkmvz1999} the existence of $1/N$ expansion of the orthogonal polynomials, for a more general case, is shown using Riemann-Hilbert method (see, for instance, Theorem 7.10 in \cite{dkmvz1999}; also see Section 3 in \cite{ercolani2003asymptotics}).  
Using this fact, the $1/N$ expansion is obtained by simply positing the $1/N$ expansion of the orthogonal polynomial and plugging into the second order differential equation \eqref{eq_ode}.

According to \cite{dkmvz1999,ercolani2003asymptotics} we know that for $l\in\NN$ there exists the asymptotic expansion
\begin{equation*}
  p_{n+r}(z)=\left(\frac{N}{2\pi^3}\right)^{1/4}\exp\left(N T  Y_{-1}(z) + Y_0(z) + \frac{1}{N} Y_1(z) + \frac{1}{N^2}Y_2(z) +\cdots +\frac{1}{N^l} Y_l(z)+{\cal O}(N^{-l-1})\right)
\end{equation*}
uniformly in a compact subset of $\CC\setminus[-F_0, F_0]$.  (We put the prefactor $\left( \frac{N}{2\pi^3}\right)^{1/4}$ so that it matches the behavior of $\gamma_{n+r}$ in Lemma \ref{lem:rho}.) Note that the functions $Y_k$'s depend on $r$ and $T$ ($=n/N$).  
Putting the asymptotic expansion of $p_{n+r}$ into the equation \eqref{eq_ode}, we get
\begin{align*}
0&=\frac{F_0^2}{2NT}\bigg( \bigg( NTY_{-1}'(z)+ \sum_{k=0}^{l} \frac{Y_{k}'(z)}{N^k} \bigg)^2 +  NTY_{-1}''(z)+\sum_{k=0}^{l} \frac{Y_{k}''(z)}{N^k}  \bigg)
\\ &\qquad\qquad\qquad -2z \bigg( NTY_{-1}'(z)+\sum_{k=0}^{l} \frac{Y_{k}'(z)}{N^k} \bigg)+2NT+2r+{\cal O}(N^{-l-1}) .
\end{align*}
Expanding in large $N$, one gets the following equations at each order.
\begin{align*}
N^1&:\quad \frac{F_0^2 Y_{-1}'(z)^2}{2} -2z Y_{-1}'(z)+2=0,
\\
N^0&:\quad\frac{F_0^2}{2}\left( 2 Y_{-1}'(z) Y_{0}'(z) +  Y_{-1}''(z) \right)-2z\, Y_{0}'(z)+2 r=0,
\\
N^{-j}&:\quad\frac{F_0^2}{2T}\left(  2TY_{-1}'(z) Y_{j}'(z)+  Y_{j-1}''(z)+ \sum_{k=0}^{j-1} Y_{k}'(z) Y_{j-k-1}'(z) \right)-2z Y_{j}'(z)=0,~~ 1\le j\le l.
\end{align*}
One can notice that the above gives the recursive relations to find $Y_j'$ for all $j\geq -1$. The first equation gives
\begin{align}\label{eq:Y-1} 
Y_{-1}'(z)&=\frac{2z}{F_0^2} \left(1- \sqrt{1- F_0^2/z^2}\right),
\end{align}
where the sign of the square root is chosen by the asymptotic behavior: $Y_{-1}'(z)={\cal O}(1/z)$.  Using the equations at the orders $N^0$ and $N^{-1}$ we get
\begin{equation}
\label{eq:Y01}
\begin{split}
Y'_{0}(z)&=- \frac{z-z\sqrt{1-F_0^2/z^2}\,(1+2 r)}{2(z^2-F_0^2)} ,\\
Y'_{1}(z)&=\frac{4 F_0^2 z^2 \left( \Big(1-\sqrt{1-F_0^2/z^2}\Big) (1+2 r)+r(r-1) \right) +F_0^4 (1-4 r(r+1) ) }{16T z\sqrt{1-F_0^2/z^2} (z^2-F_0^2)^2 } . \end{split}
\end{equation}
Using Lemma \ref{lem:rho},
we get the following boundary behaviors.
\begin{equation*}
\begin{split}
  Y_{-1}(z) &= \log z -\frac{\ell}{2T} + {\cal O}(1/z),
  \\
Y_{0}(z) &=r \log z -\left(r+\tfrac{1}{2}\right)\log \logcap(K) + {\cal O}(1/z),
  \\
  Y_{1}(z) &=-\frac{1+6 r +6r^2}{24 T }+{\cal O}(1/z),
  \end{split}
\end{equation*}
Using these, one can obtain, by integrating \eqref{eq:Y-1} and \eqref{eq:Y01}, the followings.
\begin{align*}
Y_{-1}(z) &=\frac{z^2 (1-\sqrt{1-F_0^2/z^2})^2}{2 F_0^2}-\log\left(\frac{2z(1-\sqrt{1-F_0^2/z^2})}{F_0^2} \right)-\frac{\ell}{2T},\\ 
Y_{0}(z)&=\frac{1}{2}\log\bigg(\frac{1}{2}+\frac{1}{2\sqrt{1-F_0^2/z^2}}\bigg) + r \log \bigg(\frac{z}{2}\Big(1+ \sqrt{1-F_0^2/z^2}\Big)\bigg)-\bigg(r+\frac{1}{2}\bigg)\log \logcap(K),\\
Y_{1}(z)&=\frac{3F_0^2 \left( 2\sqrt{1-F_0^2/z^2}-1+4 r \big(1+\sqrt{1-F_0^2/z^2} + r \big) \right)-2z^2\big(1+6 r (r+1)\big) }{48 T \sqrt{1-F_0^2/z^2}(z^2-F_0^2)}. 
\end{align*}

Comparing with \eqref{eq_psi} this proves Lemma \ref{prop-WKB}.
Note that $Y_j$ depends on $r$ if $j\ge 0$, i.e. $Y_j(z)=Y_j(z;r)$, and we have $g(z)=Y_{-1}(z)+\frac{\ell}{2T}$, $\sqrt{\psi'(z)} \psi(z)^r=\ee^{Y_0(z;r)}$ and $h_r(z)=Y_1(z;r)$.

\section{Proof of three lemmas}\label{app-proof3}

\subsection{Proof of Lemma \ref{prop_psi}}

From the general relations $\kappa=\dd\arg {\bf n}(s)/\dd s$ and 
\begin{equation}\label{eq:bfn}
  {\bf n}(z_0)=\frac{\psi(z_0)}{\psi'(z_0)}|\psi'(z_0)|
\end{equation}
for a general smooth curve $\psi:\partial\DD\to\CC$ (see Section \ref{sec_density} for exaxmple), one finds that 
\begin{equation}\label{kappa-psi}
  \kappa(z_0)=|\psi'(z_0)|{\rm Re}\left(1-\frac{\psi''(z_0)\psi(z_0)}{\psi'(z_0)^2}\right)=|\psi'(z_0)|-{\rm Re}\bigg({\bf n}(z_0)\frac{\psi''(z_0)}{\psi'(z_0)}\bigg),
\end{equation}
for any boundary point $z_0\in\partial K$.

One can also find the following general relation on the boundary (we suppress the arguments):
\begin{equation}\label{kappas}
\partial_s \kappa = \big(i{\bf n}\partial -i\overline{\bf n}\overline\partial\big)  |\psi'|\,{\rm Re}\left(1-\frac{\psi''\psi}{\psi'^2}\right)
=
{\rm Im}\left({\bf n}^2\left(\frac{\psi'''}{\psi'}-\frac{3}{2}\frac{(\psi'')^2}{(\psi')^2}\right)\right).
\end{equation}

For any vectors ${\bf v}_1$ and ${\bf v}_2$ the following identities hold at $z_0\in\partial K$.
\begin{equation}\label{psi-psi} \begin{split}
\partial_{\bf v_1}\frac{|\psi'(z_0)|}{\psi(z_0)}
&= \frac{|\psi'(z_0)|}{\psi(z_0)} \left(2{\rm Re} \left( {\bf v_1} \frac{\psi''(z_0)}{2\psi'(z_0)} \right)- {\bf v_1} \frac{\psi'(z_0)}{\psi(z_0)}  \right), \\
\partial_{\bf v_2}\partial_{\bf v_1}\frac{|\psi'(z_0)|}{\psi(z_0)}
&=\frac{|\psi'(z_0)|}{\psi(z_0)}\Bigg\lgroup \left({\rm Re} \left( {\bf v_1} \frac{\psi''(z_0)}{\psi'(z_0)} \right)- {\bf v_1} \frac{\psi'(z_0)}{\psi(z_0)}  \right) \left({\rm Re} \left( {\bf v_2} \frac{\psi''(z_0)}{\psi'(z_0)} \right)- {\bf v_2} \frac{\psi'(z_0)}{\psi(z_0)}  \right) \\
&\qquad +{\rm Re}\left( {\bf v_1 v_2}\left( \frac{\psi'''(z_0)}{\psi'(z_0)} - \frac{(\psi''(z_0))^2}{(\psi'(z_0))^2}  \right) \right) + {\bf v_1 v_2} \left( \frac{(\psi'(z_0))^2}{(\psi(z_0))^2} - \frac{\psi''(z_0)}{\psi(z_0)}  \right) \Bigg\rgroup.
\end{split}
\end{equation}
In Lemma \ref{prop_psi} one can see that, for any non-negative integers $k$ and $l$,
\begin{align}
\Theta_{k,l}(z_0)&=\frac{\psi(z_0)}{|\psi'(z_0)|}  \partial_{\bf n}^k \partial_{\bf t}^l \frac{|\psi'(z_0)|}{\psi(z_0)}\nn
\end{align}
comes from the Taylor expansion.  And the uniform bound stated in the lemma comes from the real-analyticity of the function $|\psi'|/\psi$ around $\partial K$.
Explicit calculations using \eqref{kappa-psi}, \eqref{kappas} and \eqref{psi-psi} give the following list of identities that hold at the boundary $\partial K$.
\begin{align*}
\Theta_{1,0}&= -\kappa,\\
\Theta_{0,1}&=-|\psi'|\, {\rm Im}\left(  \frac{\psi\psi''}{(\psi')^2} \right)- i |\psi'| ,\\
\Theta_{2,0}&=\kappa^2+{\rm Re}\left( {\bf n}^2\left( \frac{\psi'''}{\psi'} - \frac{(\psi'')^2}{(\psi')^2}  \right) \right) + |\psi'|\kappa - i|\psi'|^2\,{\rm Im}\left(\frac{\psi\psi''}{(\psi')^2}\right),\\
\Theta_{1,1}&= |\psi'|(\kappa + |\psi'|) \, {\rm Im} \left( \frac{\psi\psi''}{(\psi')^2} \right)  -\partial_s\kappa + 2i|\psi'|\kappa,\\
\Theta_{0,2}&=|\psi'|^2  \left( {\rm Im} \left( \frac{\psi\psi''}{(\psi')^2} \right)+ i \right)^2-{\rm Re}\left( {\bf n}^2\left( \frac{\psi'''}{\psi'} - \frac{(\psi'')^2}{(\psi')^2}  \right) \right) - |\psi'|\kappa + i|\psi'|^2\,{\rm Im}\left(\frac{\psi\psi''}{(\psi')^2}\right).
\end{align*}
All the above identities are true at the boundary.  We suppress all the arguments for simplicity.

To prove Lemma \ref{prop_psi} we only need to show the following relations: 

\vspace{0.2cm}
\begin{lemma}\label{lemma_Theta}
Let $z_0\in\partial K$. We have the following identities
\begin{align}\label{eq:IM}
 {\rm Im}\left( \frac{\psi(z_0)\psi''(z_0)}{(\psi'(z_0))^2} \right) &=  -\frac{1}{3}\frac{\partial_s \kappa}{  |\psi'(z_0)| \kappa},\\
 \label{eq:Re}
{\rm Re}\left({\bf n}^2\left(  \frac{\psi'''(z_0)}{\psi'(z_0)}-\frac{(\psi''(z_0))^2}{(\psi'(z_0))^2} \right) \right)&=\kappa^2+\frac{(\partial_s \kappa)^2}{3 \kappa^2}-\frac{\partial_s^2 \kappa}{3 \kappa} -|\psi'(z_0)| \kappa. 
\end{align}
\end{lemma}

\begin{proof}
Unlike the previous relations in this section, the above two hold {\em only} for elliptical droplets, and the proof can be done using the explicit parametrization of the boundary by
\begin{equation}\label{eq:gamma}
\gamma: [0,2\pi) \to \CC ,\qquad \gamma(\theta)= \phi(\ee^{i\theta})=\sqrt{\frac{T}{1-t^2}}\left( \ee^{i \theta}+t \ee^{-i \theta}\right).
\end{equation} 
It is straightforward to calculate
\begin{align}\label{eq_d2phi}
\phi'(\ee^{i\theta})=\sqrt{\frac{T}{1-t^2}} \left(1-\frac{t}{\ee^{2i\theta}} \right), \qquad \phi''(\ee^{i\theta})=\sqrt{\frac{T}{1-t^2}} \frac{2t}{\ee^{3i\theta}},\qquad \phi'''(\ee^{i \theta})= - \sqrt{\frac{T}{1-t^2}} \frac{6t}{\ee^{4i \theta}}.
\end{align}
The following general relations are derived from $\psi(\phi(u))=u$.
\begin{equation}\label{eq_d3phi}
\psi'(\phi(u))=\frac{1}{\phi'(u)} ,\quad \psi''(\phi(u))=-\frac{\phi''(u)}{\phi'(u)^3},\quad	\psi'''(\phi(u))=-\frac{\phi'''(u)}{\phi'(u)^4}+\frac{3\phi''(u)^2}{\phi'(u)^5}.
\end{equation}
Using \eqref{eq:bfn} and the above two sets \eqref{eq_d2phi}\eqref{eq_d3phi} of equations one gets
\begin{align}\label{eqnnn}
{\bf n}(\gamma(\theta))=\frac{\ee^{i \theta}-t \ee^{-i \theta}}{\sqrt{1+t^2-2t\cos(2\theta)}}.
\end{align}
Using \eqref{eq_d2phi}\eqref{eq_d3phi}\eqref{eqnnn}, one obtains the following expressions, 
\begin{align}\label{eq:curvature-psi}
\frac{\psi(\gamma(\theta))\psi''(\gamma(\theta))}{\psi'(\gamma(\theta))}
&=\frac{2t\left(t-\ee^{-2i \theta} \right)}{|\ee^{i\theta}-t\ee^{-i \theta} |^2},
\\
{\bf n}(\gamma(\theta))^2\left(  \frac{\psi'''(\gamma(\theta))}{\psi'(\gamma(\theta))}-\frac{(\psi''(\gamma(\theta)))^2}{(\psi'(\gamma(\theta)))^2} \right) &=\frac{2t(1-t^2)(3+t\ee^{-2i\theta})(\ee^{-\theta}-t\ee^{i\theta})^2}{T\left| \ee^{\theta}-t\ee^{-i\theta} \right|^6},
\end{align}
that appear in the left-hand sides of \eqref{eq:IM} and of \eqref{eq:Re}.
Taking the imaginary and the real part respectively, one gets the following evaluation at $\gamma(\theta)\in\partial K$.
\begin{align*}
{\rm Im}\left(\frac{\psi\,\psi''}{\psi'}\right)\bigg|_{\gamma(\theta)}&=\frac{2t\sin(2\theta)}{1+t^2-2t \cos(2\theta)},
\\
{\rm Re}\left({\bf n}^2\left(  \frac{\psi'''}{\psi'}-\frac{\psi''^2}{(\psi'^2} \right) \right)\bigg|_{\gamma(\theta)}&= \frac{2t(1-t^2)(t^3 -6t+ t^2 \cos(2\theta)+ 3\cos(2\theta)+t\cos(4\theta) )}{T\left| 1+t^2-2t\cos(2\theta) \right|^3}.
\end{align*}
The proof will complete if we can show the following relations:
\begin{align}\label{eq:kappa0}
\kappa(\gamma(\theta))&= \frac{\left(1-t^2 \right)^{3/2}}{\sqrt{T} \left(1+t^2-2t\cos(2\theta) \right)^{3/2}}, \\
\label{eq:kappa1}
\partial_s \kappa(\gamma(\theta))&= -\frac{6(1-t^2)^2 t \sin(2 \theta)}{T \left(1+t^2-2t\cos(2\theta) \right)^{3} }, \\
\label{eq:kappa2}
\partial_s^2 \kappa(\gamma(\theta))&= -\frac{12(1-t^2)^{5/2} t \left( 2t(\cos(4\theta)-2)+(1+t^2)\cos(2\theta) \right)  }{T^{3/2}  \left(1+t^2-2t\cos(2\theta) \right)^{9/2}}, \\\label{eq:psiprime}
|\psi'(\gamma(\theta))|&= \sqrt{\frac{1-t^2}{T}}\frac{1}{\sqrt{1+t^2-2t\cos(2\theta)}},
\end{align}
that appear in the right-hand sides of \eqref{eq:IM} and \eqref{eq:Re}.  We remind that the function $\gamma$ is defined at \eqref{eq:gamma}.

The first relation \eqref{eq:kappa0} is obtained using \eqref{eq:curvature-psi} and \eqref{kappa-psi}.  The second \eqref{eq:kappa1} and the third \eqref{eq:kappa2} are obtained by the direct calculation using
\begin{equation*}
	\partial_s\kappa(\gamma(\theta))=\frac{1}{|\phi'(\ee^{i\theta})|}\frac{\dd\kappa(\gamma(\theta))}{\dd\theta},
	\qquad
		\partial_s^2\kappa(\gamma(\theta))=\frac{1}{|\phi'(\ee^{i\theta})|}\frac{\dd(\partial_s\kappa(\gamma(\theta)))}{\dd\theta}.
\end{equation*}
The equation \eqref{eq:psiprime} is simple to obtain using \eqref{eq_d2phi}\eqref{eq_d3phi}.
\end{proof}

\subsection{Proof of Lemma \ref{lemma_hn_relations}}

Using \eqref{eq:hr} one obtains that
\begin{equation*}
\begin{split}
h_0(\gamma(\theta))+h_{-1}(\gamma(\theta))&=-\frac{ 3F_0^2+2z^2 }{24T  \sqrt{1-F_0^2/z^2} \, (z^2-F_0^2)}\bigg|_{z=\gamma(\theta)}=-\frac{\left( \ee^{i\theta}+t \ee^{-i\theta} \right)  \left( \ee^{2i\theta}+t^2 \ee^{-2i\theta}+8t \right)}{12T \left( \ee^{i\theta}-t \ee^{-i\theta} \right)^3}.
\end{split}
\end{equation*}
Taking the real part, this gives
\begin{align}\label{eq_id_2part}
{\rm Re}\left( h_0(\gamma(\theta))+h_{-1}(\gamma(\theta)) \right)&=-\frac{\left( 1-t^2 \right)  \left( 1 - 26 t^2 + t^4 + 6 t \cos(2 \theta) + 6 t^3 \cos(2 \theta) + 12 t^2 \cos(4 \theta) \right)}{12T \left( 1+t^2-2t\cos(2\theta) \right)^3},
\end{align}
which is the left-hand side of the first equation to prove.
Using \eqref{eq:kappa0}\eqref{eq:kappa1} one can confirm the first equation of Lemma \ref{lemma_hn_relations}.

Using \eqref{eq:hr} one obtains that
\begin{align*}
h_{-1}(\gamma(\theta))-h_0(\gamma(\theta))=-\frac{F_0^2}{4T(z^2-F_0^2)}\bigg|_{z=\gamma(\theta)}=-\frac{t\left(\ee^{-2i\theta}+t^2\ee^{2i\theta}-2t\right)}{T\left| \ee^{i\theta}-t\ee^{-i\theta} \right|^4}.
\end{align*}
Using \eqref{eq:kappa0}\eqref{eq:kappa1}\eqref{eq:psiprime} one can confirm the second equation of Lemma \ref{lemma_hn_relations}.

\subsection{Proof of Lemma \ref{lemma_Upsilon}}

Using \eqref{eqnnn}\eqref{eq:kappa0}\eqref{eq:kappa1}\eqref{eq:psiprime}
 we get
\begin{align*}
\frac{\sqrt{T}(\cc{{\bf n}} t-{\bf n})}{\sqrt{1-t^2}} \frac{|\psi'(z_0)|}{\psi(z_0)}&=-\frac{1+t^2-2t \ee^{-2i\theta}}{1+t^2-2t\cos(2\theta)} =-1 + \frac{i \partial_s \kappa}{3|\psi'(z_0)| \kappa},  \\
\frac{\sqrt{T}(\cc{{\bf n}} t+{\bf n})}{\sqrt{1-t^2}} \frac{|\psi'(z_0)|}{\psi(z_0)}&= \frac{1-t^2}{1+t^2-2t\cos(2\theta)} =\frac{\kappa}{|\psi'(z_0)|}.
\end{align*}

\section*{Acknowledgements}
R. Riser acknowledges the support by FWO research grant G.0934.013, and Belgian Interuniversity Attraction Pole P7/18.  This work was supported, in part, by the University of South Florida Research \& Innovation Internal Awards Program under Grant No. 0087253.
The project began at ``Conference on Integrable Systems, Random Matrix Theory, and Combinatorics'' that held at the University of Arizona in 2013.  We thank Robert Buckingham, Maurice Duits, Arno Kuijlaars, Kenneth McLaughlin and Thorsten Neuschel for useful discussions.

\bibliography{bib}

\begin{thebibliography}{10}

\bibitem{agam2002viscous}
O.~Agam, E.~Bettelheim, P.~Wiegmann, and A.~Zabrodin.
\newblock Viscous fingering and the shape of an electronic droplet in the
  quantum hall regime.
\newblock {\em Physical review letters}, 88(23):236801, 2002.

\bibitem{ameur2011}
Y.~Ameur, H.~Hedenmalm, and N.~Makarov.
\newblock Fluctuations of eigenvalues of random normal matrices.
\newblock {\em Duke Mathematical Journal}, 159(1):31--81, 07 2011.

\bibitem{AmeurKangMakarov}
Y.~{Ameur}, N.-G. {Kang}, and N.~{Makarov}.
\newblock Rescaling {W}ard identities in the random normal matrix model.
\newblock {\em ArXiv e-prints}, October 2014.

\bibitem{balogh2012strong}
F.~Balogh, M.~Bertola, S.-Y. Lee, and K.~T.-R. Mclaughlin.
\newblock Strong asymptotics of the orthogonal polynomials with respect to a
  measure supported on the plane.
\newblock {\em Communications on Pure and Applied Mathematics}, 68(1):112--172,
  2015.

\bibitem{bender}
M.~Bender.
\newblock Edge scaling limits for a family of non-{H}ermitian random matrix
  ensembles.
\newblock {\em Probability Theory and Related Fields}, 147(1-2):241--271, 2010.

\bibitem{Ber}
R.~J. Berman.
\newblock Bergman kernels for weighted polynomials and weighted equilibrium
  measures of $\mathbb{C}^{n}$.
\newblock {\em Indiana Univ. Math. J.}, 58:1921--1946, 2009.

\bibitem{harnad}
M.~Bertola, B.~Eynard, and J.~Harnad.
\newblock Duality, biorthogonal polynomials and multi-matrix models.
\newblock {\em Communications in Mathematical Physics}, 229(1):73--120, 2002.

\bibitem{davis1974schwarz}
P.~J. Davis.
\newblock {\em The Schwarz Function and Its Applications}.
\newblock Carus mathematical monographs. Mathematical Association of America,
  1974.

\bibitem{deift}
P.~Deift.
\newblock {\em Orthogonal Polynomials and Random Matrices: a
  {R}iemann-{H}ilbert Approach}, volume~3 of {\em Courant Lecture Notes}.
\newblock American Mathematical Society, Providence, 2000.

\bibitem{dkmvz1999}
P.~Deift, T.~Kriecherbauer, K.~T.-R. McLaughlin, S.~Venakides, and X.~Zhou.
\newblock Strong asymptotics of orthogonal polynomials with respect to
  exponential weights.
\newblock {\em Communications on Pure and Applied Mathematics},
  52(12):1491--1552, 1999.

\bibitem{dkmvz2001-v2}
P.~Deift, T.~Kriecherbauer, K.~T.-R. McLaughlin, S.~Venakides, and X.~Zhou.
\newblock Uniform asymptotics for polynomials orthogonal with respect to
  varying exponential weights and applications to universality questions in
  random matrix theory.
\newblock {\em Communications on Pure and Applied Mathematics},
  52(11):1335--1425, 1999.

\bibitem{dkmvz2001}
P.~Deift, T.~Kriecherbauer, K.~T.-R. McLaughlin, S.~Venakides, and X.~Zhou.
\newblock A {R}iemann-{H}ilbert approach to asymptotic questions for orthogonal
  polynomials.
\newblock {\em Journal of Computational and Applied Mathematics}, 133(1-2):47
  -- 63, 2001.
\newblock 5th Int. Symp. on Orthogonal Polynomials, Special Functions and their
  Applications.

\bibitem{elbau}
P.~{Elbau}.
\newblock Random normal matrices and polynomial curves.
\newblock {\em ArXiv e-prints}, July 2007.

\bibitem{felder-elbau}
P.~Elbau and G.~Felder.
\newblock Density of eigenvalues of random normal matrices.
\newblock {\em Communications in Mathematical Physics}, 259:433--450, 2005.
\newblock 10.1007/s00220-005-1372-z.

\bibitem{ercolani2003asymptotics}
N.~Ercolani and K.~T.-R. McLaughlin.
\newblock Asymptotics of the partition function for random matrices via
  {R}iemann-{H}ilbert techniques and applications to graphical enumeration.
\newblock {\em International Mathematics Research Notices}, 2003(14):755--820,
  2003.

\bibitem{forrester2010log}
P.~J. Forrester.
\newblock {\em Log-Gases and Random Matrices (LMS-34)}.
\newblock London Mathematical Society Monographs. Princeton University Press,
  2010.

\bibitem{forrester_honner}
P.~J. Forrester and G.~Honner.
\newblock Exact statistical properties of the zeros of complex random
  polynomials.
\newblock {\em Journal of Physics A: Mathematical and General}, 32(16):2961,
  1999.

\bibitem{ginibre}
J.~Ginibre.
\newblock Statistical ensembles of complex, quaternion, and real matrices.
\newblock {\em Journal of Mathematical Physics}, 6(3):440--449, 1965.

\bibitem{Hedenmalm01042013}
H.~Hedenmalm and N.~Makarov.
\newblock {C}oulomb gas ensembles and {L}aplacian growth.
\newblock {\em Proceedings of the London Mathematical Society},
  106(4):859--907, 2013.

\bibitem{hille}
E.~Hille.
\newblock A class of reciprocal functions.
\newblock {\em The Annals of Mathematics}, 27(4):pp. 427--464, 1926.

\bibitem{its-takhtajan}
A.~R. {Its} and L.~A. {Takhtajan}.
\newblock {Normal matrix models, dbar-problem, and orthogonal polynomials on
  the complex plane}.
\newblock {\em ArXiv e-prints}, August 2007.

\bibitem{mehta2004random}
Madan~Lal Mehta.
\newblock {\em Random Matrices}, volume 142.
\newblock Academic Press, 2004.

\bibitem{neuschel}
T.~Neuschel.
\newblock A uniform version of {L}aplace's method for contour integrals.
\newblock {\em Analysis International mathematical journal of analysis and its
  applications}, 32(2):121--136, 2012.

\bibitem{plancherel-rotach}
M.~Plancherel and W.~Rotach.
\newblock Sur les valeurs asymptotiques des polynomes d'{H}ermite.
\newblock {\em Commentarii Mathematici Helvetici}, 1(1):227--254, December
  1929.

\bibitem{rider2003limit}
B.~Rider.
\newblock A limit theorem at the edge of a non-{H}ermitian random matrix
  ensemble.
\newblock {\em Journal of Physics A: Mathematical and General}, 36(12):3401,
  2003.

\bibitem{rider_virag}
B.~{Rider} and B.~{Vir\'ag}.
\newblock The noise in the circular law and the {G}aussian free field.
\newblock {\em {IMRN. International Mathematics Research Notices}}, 2007(2):32,
  2007.

\bibitem{rr1312.0068}
R.~{Riser}.
\newblock Universality in {G}aussian random normal matrices.
\newblock {\em ArXiv e-prints}, November 2013.

\bibitem{Teodorescu2005407}
R.~Teodorescu, E.~Bettelheim, O.~Agam, A.~Zabrodin, and P.~Wiegmann.
\newblock Normal random matrix ensemble as a growth problem.
\newblock {\em Nuclear Physics B}, 704(3):407 -- 444, 2005.

\bibitem{eijndhoven}
S.~J.~L. van Eijndhoven and J.~L.~H. Meyers.
\newblock New orthogonality relations for the {H}ermite polynomials and related
  {H}ilbert spaces.
\newblock {\em Journal of Mathematical Analysis and Applications}, 146(1):89 --
  98, 1990.

\bibitem{wong}
R.~Wong and Y.~Zhao.
\newblock Asymptotics of orthogonal polynomials via the {R}iemann-{H}ilbert
  approach.
\newblock {\em Acta Mathematica Scientia}, 29(4):1005 -- 1034, 2009.

\end{thebibliography}
\bibliographystyle{plain}

\end{document}